\documentclass{IEEEtran}%[draftcls]
\usepackage[english]{babel}
\usepackage{amsmath,amsthm}
\usepackage{amsfonts}
\usepackage{extarrows}
\usepackage{amssymb}
\usepackage{dblfloatfix}
\usepackage[final,hyperfootnotes=false]{hyperref}
\usepackage[final]{graphicx}
\usepackage{tikz}
\usetikzlibrary{arrows,positioning,fit,backgrounds}
\usetikzlibrary{calc}
\newtheorem{theorem}{Theorem}
\newtheorem{example}{Example}
\newtheorem{corollary}{Corollary}
\newtheorem{lemma}{Lemma}

\theoremstyle{definition}
\newtheorem{defn}{Definition}

\theoremstyle{remark}
\newtheorem{remark}{Remark}
\newtheorem{fact}{Fact}

\newtheorem{notation}{Notation}
\newtheorem{property}{Property}

\DeclareMathOperator*{\esssup}{ess\,sup}
\DeclareMathOperator*{\tr}{tr}

\newcommand{\mb}{\mathbf}

\newcommand\blfootnote[1]{%
  \begingroup
  \renewcommand\thefootnote{}\footnote{#1}%
  \addtocounter{footnote}{-1}%
  \endgroup
}
\begin{document}
\title{Key Capacity for Product Sources with Application to Stationary Gaussian Processes}
\author{\IEEEauthorblockN{Jingbo Liu~~~~~~~~Paul Cuff~~~~~~~~Sergio Verd\'{u}\\}
\IEEEauthorblockA{Dept. of Electrical Eng., Princeton University, NJ 08544\\
\{jingbo,cuff,verdu\}@princeton.edu}}
\maketitle
\begin{abstract}
We show that for product sources, rate splitting
is optimal for secret key agreement using limited one-way
communication between two terminals. 
This yields an alternative
information-theoretic-converse-style proof of the tensorization property of a strong data processing
inequality originally studied by Erkip and Cover and amended
recently by Anantharam et al.
We derive a water-filling solution
of the communication-rate--key-rate tradeoff for a wide class of
discrete memoryless vector Gaussian sources which subsumes the
case without an eavesdropper. 
Moreover, we derive an explicit
formula for the maximum secret key per bit of communication
for all discrete memoryless vector Gaussian sources using a
tensorization property and a variation on the enhanced channel technique of Weingarten et al. 
Finally, a one-shot information spectrum
achievability bound for key generation is proved from which
we characterize the communication-rate--key-rate tradeoff for
stationary Gaussian processes.
\end{abstract}


\blfootnote{This paper was presented in part at 2014 IEEE International Symposium on Information Theory (ISIT). Copyright (c) 2014 IEEE. Personal use of this material is permitted.  However, permission to use this material for any other purposes must be obtained from the IEEE by sending a request to pubs-permissions@ieee.org.}
%\IEEEpeerreviewmaketitle
\section{Introduction}\label{sec1}
%\begin{figure}
%\centering
%  % Requires \usepackage{graphicx}
%  \includegraphics[width=3in]{figure_scheme.pdf}\\
%  \caption{The basic model for secret key agreement between two terminals A and B allowing public communication from A to B.}  \label{f1}
%\end{figure}
\begin{figure}[h!]
  \centering
\begin{tikzpicture}
[node distance=1cm,minimum height=9mm,minimum width=14mm,arw/.style={->,>=stealth'}]
  \node[coordinate] (p1) {};
  \node[rectangle,draw,rounded corners] (C) [below =of p1] {C: Eavesdropper};
  \node[rectangle,draw,rounded corners] (A) [left =2cm of p1] {A};
  \node[rectangle,draw,rounded corners] (B) [right =2cm of p1] {B};
  \node[coordinate] (X) [above =0.5cm of A] {};
  \node[coordinate] (Y) [above =0.5cm of B] {};
  \node[coordinate] (K) [below =0.5cm of A] {};
  \node[coordinate] (Kh) [below =0.5cm of B] {};
  \node[coordinate] (Z) [below =0.5cm of C] {};

  \draw [arw] (X) to node[midway,above]{$X^n$} (A);
  \draw [arw] (Y) to node[midway,above]{$Y^n$} (B);
  \draw [arw] (A) to node[midway,below]{$K$} (K);
  \draw [arw] (B) to node[midway,below]{$\hat{K}$} (Kh);
  \draw [arw] (Z) to node[midway,below]{$Z^n$} (C);
  \draw [arw] (A) to node[midway,above]{$W(X^n)$} (B);
  \draw [arw] (p1) to node{} (C);
\end{tikzpicture}
\caption{The basic model for secret key agreement between two terminals A and B allowing public communication from A to B.}
\label{f1}
\end{figure}
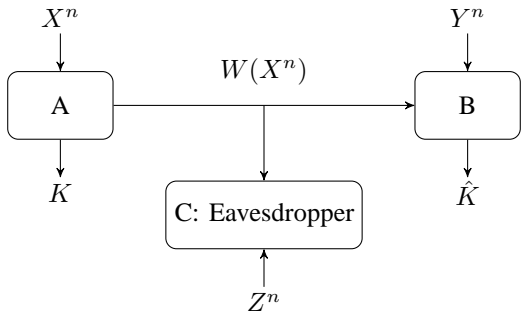
An important scenario for secret key agreement (a.k.a. key generation) arises when terminals at distant locations have access to correlated sources and are allowed to communicate publicly in order to decide on a key which is kept unknown to an eavesdropper.

The fundamental limit on the amount of secret key that can be generated from discrete memoryless sources was studied in \cite{csiszar2000common},\cite{ahlswede1998common}, where single-letter solutions were derived for the class of protocols allowing limited one-way communication from one terminal to the other. However, for many models of interest in practice, the key capacity remains unknown, since the optimizations over auxiliary random variables in those single-letter formulas are usually hard to solve.

In \cite{watanabe} the fundamental limit was extended to sources with continuous alphabets; and it was shown that for vector Gaussian sources, which are natural models of multiple input multiple output (MIMO) systems, one auxiliary random variable suffices to characterize the rate region, instead of two in the general case \cite{csiszar2000common}, and it is enough to consider auxiliary random vectors that are jointly Gaussian with the sources. This observation is formally stated in Fact~\ref{fact1} ahead, the proof of which in \cite{watanabe} was based on the enhancement technique introduced by Weingarten et~al.~\cite{weingarten2004capacity}.
Consequently, the capacity region for vector Gaussian sources was posed as a (generally non-convex) matrix optimization problem. Still, an explicit formula for the key capacity was not derived except for scalar Gaussian sources.

In this paper we provide an explicit formula for the key capacity of vector Gaussian sources by considering a more general setup: the key capacity of arbitrary product sources. Specifically, suppose terminals A and B and an eavesdropper observe discrete memoryless vector sources $\mathbf{X}=(X_i)_{i=1}^L$, $\mathbf{Y}=(Y_i)_{i=1}^L$ and $\mathbf{Z}=(Z_i)_{i=1}^L$ respectively, where
\begin{align}
P_{\mathbf{XY}}=&\prod_{i=1}^L P_{X_iY_i},\label{assump1}
\\
P_{\mathbf{XZ}}=&\prod_{i=1}^L P_{X_iZ_i}.\label{assump2}
\end{align}
We call $(\mb{X},\mb{Y},\mb{Z})$ a \emph{product source} because of the structure of its joint probability distribution. An example of product sources is illustrated in Figure~\ref{fig1}.\footnote{Actually Figure~\ref{fig1} only illustrates an unnecessarily special case of \eqref{assump1} and \eqref{assump2} where $P_{\bf XYZ}=\prod_{i=1}^L P_{X_iY_iZ_i}$; c.f.~Section~\ref{s2c}.}
%\begin{figure}
%  \centering
%  % Requires \usepackage{graphicx}
%  \includegraphics[width=3in]{figure_parallel}\\
%  \caption{An illustration of the product sources in \eqref{assump1} and \eqref{assump2}}\label{fig1}
%\end{figure}
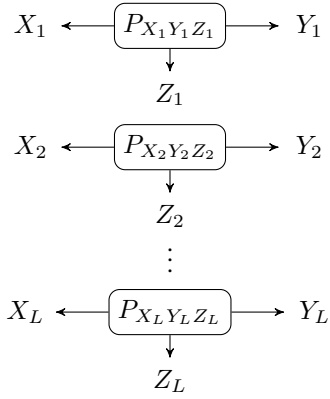
\begin{figure}[h!]
  \centering
\begin{tikzpicture}
[node distance=0.7cm,minimum height=6mm,minimum width=8mm,arw/.style={->,>=stealth'}]
  \node[rectangle,draw,rounded corners] (p1) {$P_{X_1Y_1Z_1}$};
  \node[rectangle] (x1) [left =of p1] {$X_1$};
  \node[rectangle] (y1) [right =of p1] {$Y_1$};
  \node[rectangle] (z1) [below =0.3cm of p1] {$Z_1$};
  \draw [arw] (p1) to node{} (x1);
  \draw [arw] (p1) to node{} (y1);
  \draw [arw] (p1) to node{} (z1);

  \node[rectangle,draw,rounded corners] (p2) [below =1cm of p1] {$P_{X_2Y_2Z_2}$};
  \node[rectangle] (x2) [left =of p2] {$X_2$};
  \node[rectangle] (y2) [right =of p2] {$Y_2$};
  \node[rectangle] (z2) [below =0.3cm of p2] {$Z_2$};
  \draw [arw] (p2) to node{} (x2);
  \draw [arw] (p2) to node{} (y2);
  \draw [arw] (p2) to node{} (z2);

  \node[rectangle] (dots) [below =0.7cm of p2] {$\vdots$};

  \node[rectangle,draw,rounded corners] (pL) [below =0.1cm of dots] {$P_{X_LY_LZ_L}$};
  \node[rectangle] (xL) [left =of pL] {$X_L$};
  \node[rectangle] (yL) [right =of pL] {$Y_L$};
  \node[rectangle] (zL) [below =0.3cm of pL] {$Z_L$};
  \draw [arw] (pL) to node{} (xL);
  \draw [arw] (pL) to node{} (yL);
  \draw [arw] (pL) to node{} (zL);
\end{tikzpicture}
  \caption{An illustration of the product sources in \eqref{assump1} and \eqref{assump2}.}\label{fig1}
\end{figure}
The maximal rate of secret key achievable as a function of public communication rate $r$ from A to B is denoted as $R(r)$. We show that
\begin{align}\label{disp}
R(r)=\max_{\sum_{i=1}^L r_i\le r}\sum_{i=1}^L R_{i}(r_{i}),
\end{align}
where $R_{i}(r_{i})$ is the key-communication function corresponding to the $i$-th source triple: $(X_i,Y_i,Z_i)$. This is analogous to a result due to Shannon \cite{shannon1959coding} on the rate distortion function of a product source with a separable distortion measure, which is obtained by summing the rates and distortions of points in the individual rate-distortion curves with the same slope.

In the case of jointly vector Gaussian sources without an eavesdropper (or with an eavesdropper but under a certain commutative condition on the covariance matrices), one can always apply separate invertible linear transforms on the vectors observed at A and B so that the source distribution is of the form in \eqref{assump1} and \eqref{assump2}, thus deriving an explicit formula of $R(r)$ utilizing corresponding results of scalar Gaussian sources. The solution displays a ``water filling'' behavior similar to the rate distortion function of vector Gaussian sources (e.g.\ \cite{cover2012elements}).

When the eavesdropper is present, the key-communication function is not always explicitly derived for vector Gaussian sources since the aforementioned commutative condition does not always hold. This motivates us to consider the maximum amount of secret key obtainable per bit of communication, denoted by $\eta_Z(X;Y)$. For vector Gaussian sources $\eta_{\mb{Z}}(\mb{X};\mb{Y})$ can always be explicitly found; and in order to upper bound $\eta_{\mb{Z}}(\mb{X};\mb{Y})$ we use an idea similar to but different than the \emph{enhanced channel} introduced in \cite{weingarten2004capacity}. Analogous to $\eta_Z(X;Y)$ is the notion of channel capacity per unit cost, introduced in \cite{verdu1990channel}.
As in the case of channel capacity per unit cost \cite{verdu1990channel}, a general formula for $\eta_Z(X;Y)$ can be obtained which is usually easier to compute both numerically and analytically. Some other general properties of $\eta_Z(X;Y)$ are discussed, including a formula of this quantity for product sources.

There is a curious connection between our results for product sources and the tensorization property of a strong data processing inequality originally studied by Erkip and Cover \cite{erkip1998efficiency} and amended recently by Anantharam et al. \cite{anantharam2013maximal}. Suppose $P_{XY}=P_XP_{Y|X}$ is given, and
\begin{align}
s^*(X;Y)=\sup_{U-X-Y,I(U;X)\neq0}\frac{I(U;Y)}{I(U;X)}.
\end{align}
In \cite{erkip1998efficiency} it was mistakenly claimed that
\begin{align}\label{star}
s^*(X;Y)=\rho_{\rm m}^2(X;Y)
\end{align}
where $\rho_{\rm m}^2(X;Y)$ denotes the maximal correlation coefficient \cite{witsenhausen1975sequences}. In fact, \cite{anantharam2013maximal} shows that \eqref{star} does not hold in general and gives a general but less explicit expression:
\begin{align}\label{ssup}
s^*(X;Y)=\sup_{Q_X\neq P_X}\frac{D(Q_Y||P_Y)}{D(Q_X||P_X)}
\end{align}
where\footnote{This notation defines a measure $Q_Y$ via $Q_Y(A):=\int P_{Y|X=x}(A){\rm d}Q_X(x)$ for any measurable $A\subseteq \mathcal{Y}$.}
$Q_X\to P_{Y|X} \to Q_Y$. Although $\rho_{\rm m}^2(X;Y)$ and $s^*(X;Y)$ tensorize and do agree for some simple distributions of $P_{XY}$ such as Gaussian and binary with equiprobable marginals,
it was already shown in \cite{ahlswede1976spreading} that they are not equal in general. Moreover, they are both closely linked to the problem of key generation \cite{witsenhausen1975sequences}\cite{zhao}.\footnote{For the reason we just discussed, the $\rho^2_{\rm m}(X;Y)$ in the expressions of efficiency functions in \cite{zhao} should be replaced by $s^*(X;Y)$.} To add one more connection between $s^*(X;Y)$ and key generation, we demonstrate that (\ref{disp}) implies the tensorization property of $s^*(X;Y)$. \footnote{Following our ISIT presentation of this work \cite{liu2014isit}, Beigi and Gohari \cite{beigi2015} extended such an idea and introduced several new tensorizing measures of correlation from the operational perspectives of coding theorems.}
%In fact, $s^*(X;Y)$ was originally considered by Ahlswede and G\'{a}cs \cite{ahlswede1976spreading}, where its tensorization property was proved in a very different way via the tensorization of hypercontractivity of Markov operators.
The tensorization property of $s^*(X;Y)$ turns out to be the key to many of its applications, c.f.~\cite{ahlswede1976spreading} \cite{courtade2013outer} \cite{anan_conj}.
In particular, it was shown in \cite{ahlswede1976spreading} via the tensorization of hypercontractivity of Markov operators.

Related to (memoryless) product Gaussian sources are (scalar) stationary Gaussian processes which generally have memory, since intuitively one can consider the spectral representation of stationary Gaussian processes and apply the insights from the above results concerning product sources. However there are several technical difficulties in turning this intuition into a formal proof; for example the known achievability bounds for the model under our consideration are mostly confined to memoryless sources. Thus as the first step of our proof we derive an original one-shot achievability bound via resolvability for general sources. It is relatively well known that resolvability can be applied to wiretap channels (see \cite{hayashi2006general} and the references therein), and wiretap channel codes can be employed in the encoding schemes in key agreement (an idea due to \cite{maurer1993secret}; see also \cite[Section~22.4.3]{el2011network}). Based on these connections, a recent paper \cite{bloch2013strong} derived upper and lower bounds on the key capacity for sources with memory. However those bounds may be loose, and they are still asymptotic (expressed in terms of probabilistic $\limsup$ of random variables) rather than one-shot. Moreover the setting therein is a special case of ours where the public communication rate is unlimited, and the proof technique involving modulo sums only applies to discrete sources, therefore those results are still not quite useful for resolving the achievable region for stationary Gaussian sources. In contrast, our achievability bound overcomes those issues by employing a different encoding scheme called \emph{likelihood encoder} proposed recently in \cite{song}. We then apply certain asymptotic approximation theorems for Toeplitz matrices when specializing to Gaussian processes.

\emph{Organization. }The formal definition of the key generation problem with limited one-way communication, as well as the setup of product sources and stationary Gaussian sources, are presented in Section \ref{sec2}. The main results are given in Section~\ref{sec3}. Section~\ref{sec3a} gives the central result concerning key generation from general product sources and it analyzes the special case of product Gaussian sources culminating in the ``water-filling'' solution. The necessary and sufficient condition under which general vector Gaussian sources can be converted to product Gaussian sources is also identified. Section~\ref{s3a} begins with several general properties on the maximal secret key per bit of communication, and ends with a formula for this quantity for general vector Gaussian sources which may not be convertible to product sources. Section~\ref{sec3c} presents the water-filling solution for the key-communication tradeoff for stationary Gaussian processes (Theorem~\ref{thm5}) and discusses the intuition behind it. To prove Theorem~\ref{thm5}, we derive a general one-shot achievability bound for key generation from sources with memory in Section~\ref{sec4}, and then apply it in Section~\ref{sec5} to finish the achievability proof of Theorem~\ref{thm5}. In Section~\ref{sec6} we mention some related problems involving product sources/channels.

%We shall discuss the implications of the tensorization result in the key generation setting, and simplify the efficiency function $\sup \frac{R_k(R_p)}{R_p}$ to a form similar to (\ref{ssup}), which is usually easier to optimize than the expression involving auxiliary random variables.

\section{Preliminaries}\label{sec2}
\subsection{Key Generation with One-Way Communication: Basic Setup}\label{sec2b}
Throughout this paper, random variables (but not excluding deterministic constants) are denoted by upper-case letters, and vectors and matrices are denoted in bold face.

Consider the source model illustrated in Figure~\ref{f1}. Stationary sources of blocklength $n$ have the joint distribution $P_{X^nY^nZ^n}$, where $X_i^j$ is a short hand notation for $(X_i,\dots,X_j)^{\top}$ and $X^n:=X_1^n$. Upon receiving $X^n$, terminal A computes an integer $K\in\mathcal{K}$ and a message $W\in\mathcal{W}$,
possibly stochastically\footnote{Here we allow stochastic encoders to be consistent with the achievability scheme in Section~\ref{sec4}, although in the literature $K$ and $W$ have often been defined as functions of $X^n$.}, according to $P_{KW|X^n}$. The message $W$ is then sent
through a noiseless public channel to terminal B, and B computes the key $\hat{K}=\hat{K}(W(X^n),Y^n)\in\mathcal{K}$ based on its available information. The probability of error and the measure of security are defined by
\begin{align}
\epsilon_n&=\mathbb{P}[K\neq \hat{K}],\label{e4}
\\
\nu_n&=\log|\mathcal{K}|-H(K|W,Z^n).
\end{align}

A rate pair $(R,r)$ is said to be achievable if a sequence of schemes can be designed to satisfy the following conditions on the probability of disagreement and security:
\begin{align}
\limsup_{n\to\infty}\frac{1}{n}\log|\mathcal{W}|&\le r,\\
\liminf_{n\to\infty}\frac{1}{n}\log|\mathcal{K}|&\ge R,
\\
\lim_{n\to\infty}\epsilon_n&=0,\\
\lim_{n\to\infty}\nu_n &=0.\label{e9}
\end{align}
In the remainder of Section~\ref{sec2b} we focus on the case of stationary memoryless sources with per-symbol distribution $P_{XYZ}$. The achievable rate region is defined as
\begin{align}
\mathcal{R}(X,Y,Z):=\{(R,r)\colon(R,r)\textrm{ is achievable}\},
\end{align}
and the key-communication function
\begin{align}
R(r):=\sup\{R\colon(R,r)\in\mathcal{R}(X,Y,Z)\}
\end{align}
characterizes the maximal possible key rate given a certain public communication rate.

From \cite{csiszar2000common}, the region $\mathcal{R}(X,Y,Z)$ is the union of
\begin{align}\label{region0}
[0,I(V;Y|U)-I(V;Z|U)]\times[I(U,V;X)-I(U,V;Y),\infty)
\end{align}
over all $U,V$ such that $(U,V)-X-(Y,Z)$.

For key generation with one-way communication under our consideration, only $P_{XY}$ and $P_{XZ}$ affect the achievable key-communication rates. Although beyond those joint distributions we do not need further information about the source, it is customary to say that $P_{XYZ}$ is stochastically degraded~\cite{cover2012elements} if $X-Y-Z$ form a Markov chain under a joint distribution whose pairwise distributions are $P_{XY}$ and $P_{XZ}$. In this case, the above region can be simplified to the union of
\begin{align}\label{region0}
[0,I(V;Y)-I(V;Z)]\times[I(V;X)-I(V;Y),\infty)
\end{align}
over all $V$ such that $V-X-(Y,Z)$.

For jointly Gaussian vectors $(\mb{X},\mb{Y},\mb{Z})$ it is generally not true that $P_{\mb{X}\mb{Y}\mb{Z}}$ is stochastically degraded. Thus it might seem remarkable that still only one auxiliary random variable is needed; and it can be chosen to be jointly Gaussian with the source vectors, as summarized below:
\begin{fact}[\cite{watanabe}]\label{fact1}
Suppose $X^L,Y^L$, and $Z^L$ are jointly Gaussian vectors of length $L$, and $U$ and $V$ are random variables such that $(U,V)-X^L-(Y^L,Z^L)$ form a Markov chain. Then there exists a random vector $\bar{U}^L$ in $\mathbb{R}^L$ such that $\bar{U}^L,X^L$ are jointly Gaussian, $\bar{U}^L-X^L-(Y^L,Z^L)$, and
\begin{align}
I(\bar{U}^L;X^L)-I(\bar{U}^L;Y^L)&\le I(U,V;X^L)-I(U,V;Y^L),
\\
I(\bar{U}^L;Y^L)-I(\bar{U}^L;Z^L)&\ge I(V;Y^L|U)-I(V;Z^L|U).
\end{align}
\end{fact}
As a consequence of Fact~\ref{fact1} the region $\mathcal{R}(\bf X,Y,Z)$ is the union of
\begin{align}\label{region1}
[0,I(\mb{U};\mb{Y})-I(\mb{U};\mb{Z})]\times[I(\mb{U};\mb{X})-I(\mb{U};\mb{Y}),\infty)
\end{align}
over all $\bf U$ such that $\bf U-X-(Y,Z)$ and $\bf U,X$ are jointly Gaussian. Note that $(\mb{U},\mb{X},\mb{Y},\mb{Z})$ are necessarily jointly Gaussian as well because of the Markov chain condition.

\subsection{Key Generation from Product Sources}\label{s2c}
A product source is just a particular stationary memoryless source in which $P_{\mb{X}\mb{Y}\mb{Z}}$ has the structure of \eqref{assump1} and \eqref{assump2}. Hence the setup for a product source model is the same as the stationary case of Part \ref{sec2b} with the exception that $X,Y$ and $Z$ are replaced with $L$-vectors $\bf X,Y$ and $\mb{Z}$.

We remind the reader that $\mathcal{R}(X,Y,Z)$ in \ref{sec2b} depends only on $P_{XY}$ and $P_{XZ}$, hence we do not need to define a product source with the more stringent condition of $P_{\bf XYZ}=\prod_{i=1}^L P_{X_iY_iZ_i}$.

\section{Main Results}\label{sec3}

\subsection{Secret Key Generation from Product Sources}\label{sec3a}
Suppose we know the function $R_i(r)$ for each ``factor'' in the product source; what can we say about $R(r)$ for the whole source?
As Theorem \ref{thm3} elucidates, the rate splitting approach in which we produce keys separately for each factor source (with appropriately selected rates) achieves the optimal key rate. This is analogous to a result in rate distortion theory \cite{shannon1959coding} as remarked in the introduction.
\begin{theorem}%[Rate splitting is optimal for secret key generation with product sources]
\label{thm3}
In the problem of key generation from product sources satisfying \eqref{assump1} and \eqref{assump2}, the maximum key rate satisfies
\begin{align}\label{eq18}
R(r)=\max_{\sum_{i=1}^L r_i\le r}\sum_{i=1}^L R_{i}(r_{i}),
\end{align}
where $R_{i}(r_i)$ is the key-communication function corresponding to the $i$'th source triple $(X_i,Y_i,Z_i)$. Further, if $R_{i}(\cdot)$ is differentiable and $(r_1^*,\dots,r_L^*)$ achieves the maximum in \eqref{eq18},
then for each $i$, either $R'_{i}(r^*_{i})=\mu$ for some constant $\mu$ or $r^*_{i}=0$.
\end{theorem}
\begin{remark}
The result of \eqref{eq18} can be equivalently expressed as $\mathcal{R}({\bf X,Y,Z})=\sum_{i=1}^L\mathcal{R}(X_i,Y_i,Z_i)$, where the summation is the Minkowski sum of sets in the Euclidean space.
\end{remark}
\begin{proof}
Each key rate of $R_{i}(r^*_{i})$ can be approached by a scheme that operates on the $i$'th source triple separately using a communication rate of $r_i^*$. From \eqref{assump2}, the combination of these schemes forms a legitimate scheme for the product source, since the keys generated by those schemes are independent and their combination is (asymptotically) independent of $W$ and $\mb{Z}^n$. Thus $\geq$ holds in \eqref{eq18}~
\footnote{From this argument we see that $\ge$ in \eqref{eq18} only requires \eqref{assump2} but not \eqref{assump1}. In words, a rate-splitting key agreement scheme designed for product sources will be reliable and secure even if the $\mb{Y}$ vector is correlated. This can only correlate the decoding errors, which are negligible anyway.
%Of course, \eqref{e48} can also be proved by using the single letter expression of the rate region, which is also simple but more tedious. We omit it here since (\ref{e48}) is the trivial part of Theorem \ref{thm3} and not our main focus.
}.

By (\ref{region0}) the achievable region $\mathcal{R}(X^L,Y^L,Z^L)$ is the union of
\begin{align}\label{region}
&[0,I(V;Y^L|U)-I(V;Z^L|U)]
\nonumber
\\
&\times[I(U,V;X^L)-I(U,V;Y^L),\infty)
\end{align}
over all $(U,V)$ such that $(U,V)-X^L-(Y^L,Z^L)$. The achievable region with rate splitting is the union of
\begin{align}
&\left[0,\sum_{i=1}^L I(V_i;Y_i|U_i)-\sum_{i=1}^LI(V_i;Z_i|U_i)\right]
\nonumber
\\
&\times\left[\sum_{i=1}^L I(U_i,V_i;X_i)-\sum_{i=1}^L I(U_i,V_i;Y_i),\infty\right)
\end{align}
over all $(U_i,V_i)$ such that $(U_i,V_i)-X_i-(Y_i,Z_i)$, which contains the union of \eqref{region}, according to Lemma~\ref{lem1} in Appendix~\ref{app1}. Hence we also have $\leq$ in \eqref{eq18}.

The last claim in the theorem for differentiable $R_{i}(\cdot)$ can be verified from the KKT condition and the fact that $R_i(\cdot)$ is a concave function for each $i$.
\end{proof}

From Theorem \ref{thm3} we derive the communication-rate--key-rate tradeoff for product Gaussian sources. The solution displays a ``water-filling'' behaviour which is reminiscent of the rate-distortion function for Gaussian vectors~\cite{cover2012elements}.
\begin{theorem}\label{thm4}
If $(X^L,Y^L,Z^L)$ are product Gaussian sources, then the achievable communication and key rates are parameterized by $\mu>0$ as
%\footnote{If $\beta_i=m_i=0$, we take $\frac{\beta_i(m_i+1)}{(\beta_i+1)m_i}=1$.}
\begin{align}\label{rp0}
r=\frac{1}{2}\sum_{i\colon\beta_i>\mu}\log\frac{\beta_i(\mu+1)}{(\beta_i+1)\mu},
\end{align}
\begin{align}\label{rk0}
R=\frac{1}{2}\sum_{i\colon\beta_i>\mu}\log\frac{\beta_i+1}{\mu+1},
\end{align}
where
\begin{align}
\beta_i:=\frac{\rho_{X_iY_i}^2-\rho_{X_iZ_i}^2}{1-\rho_{X_iY_i}^2}.
\end{align}

%$\mathcal{I}:=\{i:\beta_i>\mu\}$.
\end{theorem}

\begin{remark}
The usefulness of the $i$-th component of the product source is completely captured by $\beta_i$. In \eqref{rp0} and \eqref{rk0} the $i$'th term enters the summations if and only if $\beta_i$ is large enough; in other words, only the components that are strong enough are ``on''. This is similar to water-filling over Gaussian channels (avoiding low SNR channels) and rate-distortion (neglecting to compress weak source components).
\end{remark}
\begin{remark}
In Theorem \ref{thm4} if we drop assumption that $Y^L$ is Gaussian, i.e, only assume that $(X^L,Y^L,Z^L)$ is a product source where $X^L,Z^L$ are jointly Gaussian, then Theorem \ref{thm4} will provide an inner bound on the achievable region. To see this, let $Y^L_G$ be a random variable such that $X^L$ and $Y^L_G$ are jointly Gaussian, and $(X^L,Y^L_G)$ has the same first and second order statistics as $(X^L,Y^L)$. If $U$ is a Gaussian auxiliary random variable as in (\ref{region1}), then $(U,Y^L_G)$ has the same first and second order statistics as $(U,Y^L)$.
For arbitrary $P\ll Q$, define the \emph{relative information}
\begin{align}
\imath_{P\|Q}(x)=\log\frac{{\rm d}P}{{\rm d}Q}(x)
\end{align}
as the logarithm of the Radon-Nikodym derivative.
Then we have
\begin{align}
&\quad I(U;Y^L)-I(U;Y_G^L)\nonumber
\\
&=D(P_{UY^L}||P_U\times P_{Y^L})
-\mathbb{E}\left[\imath_{P_{UY_G^L}||P_U\times P_{Y_G^L}}(U,Y_G^L)\right]
\\
&=D(P_{UY^L}||P_U\times P_{Y^L})
-\mathbb{E}\left[\imath_{P_{UY_G^L}||P_U\times P_{Y_G^L}}(U,Y^L)\right]\label{e61}
\\
&=D(P_{UY^L}||P_{UY_G^L})
-D(P_{Y^L}||P_{Y_G^L})
\\
&\ge0
\end{align}
where (\ref{e61}) is because $\imath_{P_{UY_G^L}||P_U\times P_{Y_G^L}}(u,y^L)$ is only a second order polynomial of $(u,y^L)$. Hence the secret key can be generated more efficiently than in the Gaussian case:
\begin{align}
I(U;X^L)-I(U;Y^L)&\le I(U;X^L)-I(U;Y_G^L)\\
I(U;Y^L)-I(U;Z^L)&\ge I(U;Y_G^L)-I(U;Z^L).
\end{align}
\end{remark}

For a positive-semidefinite matrix $\bf\Sigma$, let ${\bf \Sigma}^{-1/2}$ be a positive definite matrix such that ${\bf \Sigma}^{-1/2}{\bf \Sigma\Sigma}^{-1/2}={\bf I}_r$, where ${\bf I}_{r}$ denotes the identity matrix of dimension $r={\rm rank}({\bf \Sigma})$. Also write ${\bf \Sigma}^{-1}=({\bf \Sigma}^{-1/2})^2$, which is the matrix inverse when $\bf \Sigma$ is invertible. Note that under this definition ${\bf \Sigma}^{-1/2}$ (and therefore ${\bf \Sigma}^{-1}$) may not be unique. The following fact about Gaussian distributions is useful. The proof is based on the singular value decomposition and is deferred to Appendix \ref{appc}.
\begin{lemma}\label{lem3}
For a set of vector random variables $\mathbf{X},\mathbf{Y},\mathbf{Z}$, there exist invertible linear transforms $\mathbf{X}\mapsto \bar{\mathbf{X}}$, $\mathbf{Y}\mapsto \bar{\mathbf{Y}}$, $\mathbf{Z}\mapsto \bar{\mathbf{Z}}$ such that all the five covariance matrices $\mathbf{\Sigma}_{\bar{\mathbf{X}}}$, $\mathbf{\Sigma}_{\bar{\mathbf{Y}}}$, $\mathbf{\Sigma}_{\bar{\mathbf{Z}}}$, $\mathbf{\Sigma}_{\bar{\mathbf{X}}\bar{\mathbf{Y}}}$, $\mathbf{\Sigma}_{\bar{\mathbf{X}}\bar{\mathbf{Z}}}$ are diagonalized if and only if $\mb{G}$ commutes with $\mb{H}$ where
\begin{align}
\mb{G}&:=\mathbf{\Sigma}^{-1/2}_{\mathbf{X}}\mathbf{\Sigma}_{\mathbf{X}\mathbf{Y}}\mathbf{\Sigma}^{-1}_{\mathbf{Y}}\mathbf{\Sigma}_{\mathbf{Y}\mathbf{X}}\mathbf{\Sigma}^{-1/2}_{\mathbf{X}},
\label{defm}
\\
\mb{H}&:=\mathbf{\Sigma}^{-1/2}_{\mathbf{X}}\mathbf{\Sigma}_{\mathbf{X}\mathbf{Z}}\mathbf{\Sigma}^{-1}_{\mathbf{Z}}\mathbf{\Sigma}_{\mathbf{Z}\mathbf{X}}\mathbf{\Sigma}^{-1/2}_{\mathbf{X}}.
\label{defn}
\end{align}
\end{lemma}
\begin{remark}
The linear transforms being invertible ensures $\mathcal{R}\bf(X,Y,Z)=\mathcal{R}\bf(\bar{X},\bar{Y},\bar{Z})$.
\end{remark}
\begin{remark}\label{rem8}
For Hermitian matrices $\mb{A}$ and $\mb{B}$ of the same dimensions, we write $\mb{A}\preceq\mb{B}$ if $\mb{B}-\mb{A}$ is positive-semidefinite.
From the positive definiteness of covariance matrices it's straightforward to show that $\mb{0}\preceq {\bf G}\preceq \mb{I}$ and $\mb{0}\preceq {\bf H}\preceq \mb{I}$, where $\mb{G}$ and $\mb{H}$ are defined in \eqref{defm} and \eqref{defn}. Indeed $\bf G$ and $\bf H$ take the roles of $\rho_{XY}^2$ and $\rho_{XZ}^2$ in the scalar Gaussian case.
\end{remark}
\begin{remark}\label{rem2}
If $X^n$, $Y^n$ and $Z^n$ are drawn from jointly stationary Gaussian processes, then the commutativity assumption in the lemma is satisfied approximately for $n$ large. This is due to the commutativity of convolution.
\end{remark}

\begin{corollary}\label{cor2}
If $\mathbf{X}$ and $\mathbf{Y}$ are jointly Gaussian vectors, then there exist invertible linear transforms $\mathbf{X}\mapsto \bar{\mathbf{X}}$ and $\mathbf{Y}\mapsto \bar{\mathbf{Y}}$ such that $\mathbf{\Sigma}_{\bar{\mathbf{X}}}$, $\mathbf{\Sigma}_{\bar{\mathbf{Y}}}$ and $\mathbf{\Sigma}_{\bar{\mathbf{X}}\bar{\mathbf{Y}}}$ are diagonalized.
\end{corollary}

Thanks to Corollary \ref{cor2}, the task of finding the key capacity of arbitrarily correlated Gaussian vector sources in the absence of an eavesdropper is reduced to the case of product Gaussian sources $(X^L,Y^L)$ satisfying \eqref{assump1} and \eqref{assump2}. Note that assuming $\mb{X}$ and $\mb{Y}$ have the same length does not lose generality since one can always pad zero coordinates to $\mb{X}$ and $\mb{Y}$ so that they have the same length. In the presence of an eavesdropper, it is not always possible to reduce the problem to the case of product sources, since the commutativity condition in Lemma~\ref{lem3} is not always fulfilled; we discuss its practical relevance later in \ref{sec3c}.

\begin{proof}[Proof of Theorem \ref{thm4}]
Reference \cite{watanabe} derived an explicit formula for the achievable key rate in the case of scalar Gaussian sources, which, in our notations, can be expressed as
\begin{align}\label{scalar2}
R(r)= \frac{1}{2}
\log\left(1+\beta^+-\beta^+ \exp(-2r)\right)
\end{align}
where $\beta:=\frac{\rho_{xy}^2-\rho_{xz}^2}{1-\rho_{xy}^2}$ and $\beta^+:=\max\{\beta,0\}$. The bases of $\log$ and $\exp$ in \eqref{scalar2} depend on the unit of the information rates (e.g.\ bits or nats).

Now consider the product sources, and suppose that $(r_1^*,\dots,r_L^*)$ achieves the maximum in \eqref{eq18}.
According to Theorem \ref{thm3}, either $R_{i}'(r^*_{i})=\frac{\beta^+_i \exp(-2r^*_{i})}{1+\beta^+_i-\beta^+_i \exp(-2r^*_{i})}=\mu$ or $r^*_{i}=0$ for each $i$, where $\mu$ is some constant. For fixed $\mu$, this means
\begin{align}
r^*_{i}=\max\left\{0,\frac{1}{2}\log
\frac{(1+\mu)\beta^+_i}{\mu(1+\beta^+_i)}\right\}.
\end{align}
Equivalently, we can write
\begin{align}\label{rp}
r^*_{i}=\frac{1}{2}\log\frac{\beta^+_i(m_i+1)}{(\beta^+_i+1)m_i},
\end{align}
where $m_i:=\min\{\mu,\beta^+_i\}$. The claim then follows by substituting the value of $r^*_{i}$ into (\ref{scalar2}) and applying (\ref{eq18}).
\end{proof}

\subsection{Secret Key per Bit of Communication}\label{s3a}
Fix $P_{XYZ}$. The \emph{secret key per bit of communication} is defined as
\begin{align}
\eta_Z(X;Y):=\sup_{r> 0}\frac{R(r)}{r}.
\end{align}
From the convexity of the achievable rate region one immediately sees that $\eta_Z(X;Y):=\lim_{r\to 0}\frac{R(r)}{r}$.

Define
\begin{align}
s^*_Z(X;Y)%\nonumber\\
%=&\sup_{U,V}\frac{I(V;Y|U)-I(V;Z|U)}{I(U,V;X)-I(U,V;Y)+I(V;Y|U)-I(V;Z|U)},\nonumber\\
:=\sup_{U,V}\frac{I(V;Y|U)-I(V;Z|U)}
{I(V;X|U)-I(V;Z|U)+I(U;X)-I(U;Y)},\label{ss}
\end{align}
where the supremum is over all $(U,V)$ such that $(U,V)-X-(Y,Z)$ form a Markov chain and that the denominator in (\ref{ss}) does not vanish. Note that the denominator is always nonnegative; if it vanishes for all $U,V$, then so does the numerator and we set $s^*_Z(X;Y)=0$. From \eqref{ss} and \eqref{region0} it is immediate to see how $s^*_Z(X;Y)$ is related to $\eta_Z(X;Y)$. In the special case of no eavesdropper, this is related to the result in \cite{zhao}, which uses the incorrect constant $\rho^2_{\rm m}(X;Y)$ as we mentioned earlier.

\begin{theorem}\label{thm3}
Secret key per bit of communication is linked to $s^*_Z(X;Y)$ by
\begin{align}\label{e17}
\eta_Z(X;Y)=\frac{s^*_Z(X;Y)}{1-s^*_Z(X;Y)}.
\end{align}
\end{theorem}
\begin{proof}
The characterization of $\mathcal{R}(X,Y,Z)$ is given in (\ref{region0}). Thus $\sup_{r> 0}\frac{R(r)}{r}=\frac{s^*_Z(X;Y)}{1-s^*_Z(X;Y)}$ follows immediately from the definition of $s^*_Z(X;Y)$. The claim of
$
\lim_{r\downarrow0}\frac{R(r)}{r}=\sup_{r> 0}\frac{R(r)}{r}
$
follows from the convexity of the achievable rate region.
\end{proof}
The following results provide some basic properties of $s^*_Z(X;Y)$. The rationale for defining $s^*_Z(X;Y)$ can be explained by Theorem~\ref{thm3} and 6) in Theorem~\ref{thm1}.
\begin{theorem}[Properties of $s^*_Z(X;Y)$]\label{thm1}~\\
\begin{enumerate}
\item For any $P_{XYZ}$,
\begin{align}\label{e35}
0\le s^*_Z(X;Y)\le 1.
\end{align}

\item For product sources $(X^L,Y^L,Z^L)$ as in \eqref{assump1} and \eqref{assump2},
\begin{align}\label{eq16}
s^*_{Z^L}(X^L;Y^L)=\max_{1\le i\le L}s^*_{Z_i}(X_i;Y_i).
\end{align}

\item For arbitrary $P_{XYZ}=P_XP_{YZ|X}$,
\begin{align}\label{e18}
&\quad s^*_Z(X;Y)\nonumber
\\
&=\sup_{Q_{VX}}\frac{I(\bar{V};\bar{Y})-I(\bar{V};\bar{Z})}
{I(\bar{V};\bar{X})-I(\bar{V};\bar{Z})+D(Q_X||P_X)-D(Q_Y||P_Y)},
\end{align}
where $\bar{V},\bar{X},\bar{Y},\bar{Z}$ have the joint distribution
\begin{align}\label{e19}
 P_{\bar{V}\bar{X}\bar{Y}\bar{Z}}(v,x,y,z)=Q_{VX}(v,x)P_{YZ|X}(y,z|x).
\end{align}
The supremum is over all $Q_{VX}$ such that the above denominator does not vanish.

Computation can be further simplified when the source has certain structures:\\
\item If $P_{XYZ}$ is stochastically degraded,
\begin{align}
s^*_Z(X;Y)&=\sup_{U}\frac{I(U;Y|Z)}{I(U;X|Z)}\nonumber\\
&=\sup_{Q_X}\frac{D(Q_Y||P_Y)-D(Q_Z||P_Z)}{D(Q_X||P_X)-D(Q_Z||P_Z)}\label{eq19}
\end{align}
where $Q_{XYZ}=Q_X P_{YZ|X}$.\\

\item As a special case of (\ref{eq19}), if $X=(X',Z)$, $Y=(Y',Z)$, then
\begin{align}\label{eq20}
s^*_Z(X;Y)=\esssup_{z\in\mathcal{Z}}s^*(X';Y'|Z=z),
\end{align}
where $\esssup$ denotes the essential supremum of a real valued function.
\item If $Z$ is constant, we recover $s_Z^*(X;Y)=s^*(X;Y)$, which is the best constant for the strong data processing inequality defined in \eqref{ssup}.
\end{enumerate}
\end{theorem}
\begin{proof}
See Appendix \ref{app2}.
\end{proof}
\begin{remark}
The interpretation of the tensorization of $s^*_Z(X;Y)$ in \eqref{eq16} is that, with small allowable public communication, it is always efficient to only use the best component of the product sources. Alternatively, the fact that rate splitting is optimal for product sources implies the tensorization property of $s_Z^*(X;Y)$.
\end{remark}

\begin{remark}\label{rem6}
If the source is stochastically degraded, then $s_Z(X;Y)$ can be computed from \eqref{eq19} which only requires optimizing over an auxiliary distribution $Q_X$, instead of the optimization over a family of distributions $P_{U|X}$ when computing the rate region via \eqref{region1}. Similarly for non-degraded sources, \eqref{e18} only involves optimizing over $Q_{VX}$ whereas the region rate region \eqref{region0} requires optimizing over $P_{UV|X}$. Thus, in either case, the optimization problem may be considerably reduced if one is only interested in $\eta_Z(X;Y)$ instead of the whole rate region.
\end{remark}

\begin{example}[Symmetric Bernoulli Source]\label{ex1}
Suppose $X,Y$ and $Z$ are symmetric Bernoulli random variables, with $\epsilon_{XY}:=\mathbb{P}[X\neq Y]$ and $\epsilon_{XZ}:=\mathbb{P}[X\neq Z]$ satisfying $\epsilon_{XY}\le\epsilon_{XZ}<\frac{1}{2}$. The achievable region $\mathcal{R}(X,Y,Z)$ was derived in \cite{chou2012separation}, from which one can obtain
\begin{align}\label{e23}
\eta_Z(X;Y)=\frac{(1-2\epsilon_{XY})^2-(1-2\epsilon_{XZ})^2}{1-(1-2\epsilon_{XY})^2}.
\end{align}
Since $X$, $Y$, and $Z$ are stochastically degraded, we can assume without loss of generality that $X-Y-Z$ form a Markov chain, and use \eqref{e17} and \eqref{eq19} to obtain \eqref{e23}. In this case \eqref{eq19} is supremized as $Q_X$ approaches the equiprobable distribution on $\{0,1\}$.
\end{example}

\begin{example}[Scalar Gaussian Source]\label{ex2}
Setting $L=1$ in Theorem~\ref{thm4} gives
\begin{align}
\eta_Z(X;Y)=\beta
\end{align}
where
\begin{align}
\beta:=\frac{\rho_{XY}^2-\rho_{XZ}^2}{1-\rho_{XY}^2}
\end{align}
for jointly Gaussian random variables $X$, $Y$ and $Z$ satisfying $\rho_{XZ}\le\rho_{XY}<1$. We remark that the (less trivial) direction of $\eta_Z(X;Y)\le\beta$ can also be expected from Example \ref{ex1} (whereas the proof of this direction using Theorem \ref{thm4} essentially relies on entropy power inequality buried in Fact \ref{fact1}); see Appendix~\ref{ex2_app}.
\end{example}

\begin{example}[Product Gaussian Source]
If $\bf X,Y$ and $\bf Z$ are as in Theorem \ref{thm4}, then
\begin{align}
\eta_{\mb{Z}}(\mb{X};\mb{Y})=\max_{1\le i\le L}\beta_i^+,
\end{align}
where $\beta_i$ is as in Theorem~\ref{thm4}.
\end{example}

In addition to the potential dimension reduction in numerical evaluations (see Remark \ref{rem6}), another important motivation for considering $\eta_Z(X;Y)$ is that there exist source distributions for which $\eta_Z(X;Y)$ can be computed analytically even though $\mathcal{R}(X,Y,Z)$ is not completely known, as epitomized by the case of vector Gaussian sources in Theorem \ref{cor1} below. Note that Theorem \ref{cor1} holds even when the commutativity in Lemma \ref{lem3} fails. The achievability (lower bound) part of Theorem \ref{cor1} is accomplished by choosing an appropriate sequence of $Q_{VX}$ in \eqref{e18} followed by routine computations; the converse part requires slightly more ingenuity: we construct a new source distribution $P_{XY\hat{Z}}$ satisfying $\mathcal{R}(X,Y,Z)\subseteq\mathcal{R}(X,Y,\hat{Z})$, but for which the commutativity in Lemma \ref{lem3} is fulfilled and $\eta_{\hat{Z}}(X;Y)=\eta_Z(X;Y)$. Details of the proof are relegated to Appendix~\ref{app_cor1}.
\begin{theorem}\label{cor1}
If $\mb{X}$, $\mb{Y}$ and $\mb{Z}$ in the key generation model are jointly Gaussian vectors, then
\begin{align}
\eta_{\mb{Z}}(\mb{X};\mb{Y})=\lambda_{\max}^+((\mb{G}-\mb{H})(\mb{I}-\mb{G})^{-1}),
\end{align}
where $\lambda_{\max}(\cdot)$ and $\lambda_{\min}(\cdot)$ denote the largest and smallest eigenvalues of a matrix, and recall the notation $\lambda_{\max}^+:=\max\{0,\lambda_{\max}\}$.
\end{theorem}

\subsection{Key-Communication Function for Stationary Gaussian Processes}\label{sec3c}
We now derive the key-rate--communication-rate tradeoff for stationary Gaussian processes $(\mathbb{X},\mathbb{Y},\mathbb{Z})$. In contrast to the setting of product sources since in this section we deal with sources with memory. However as mentioned in Remark \ref{rem2}, one can still apply Lemma \ref{lem3}, and in fact the linear transforms can be easily found. Let us discuss the intuitions before diving into the formal proof. As a first attempt, it is tempting to pick the Fourier transform as the invertible linear transforms in Lemma~\ref{lem3}, since it diagonalizes circulant matrices \cite{gray2006toeplitz}. However this is not an allowable choice, since the linear transforms in Lemma~\ref{lem3} are real, thereby excluding the Fourier transform. In general, complex linear transforms are not useful for the conversion to product sources, since complex Gaussian variables may not be independent even if their correlation coefficient is zero.

The Fourier transform, however, is not too far from the correct choice. If a circulant matrix is symmetric, we can also diagonalize it with the sine/cosine orthogonal matrix (to be defined soon). In general, the cross-correlations $R_{XY}$ and $R_{XZ}$ are not symmetric, so the trick is to first pass $\mathbb{Y}$ through a filter whose impulse response is $R_{XY}$\footnote{When $R_{XY}$ is strictly bandlimited, convolution with $R_{XY}$ becomes a degenerate linear transform. In this case we can use a signal $\hat{R}_{XY}$ as an alternative, where $\hat{R}_{XY}$ has full spectrum and agrees with $R_{XY}$ in the pass-band of $R_{XY}$. The final formula of key capacity however will remain unchanged.},
the correlation function between $\mathbb{X}$ and $\mathbb{Y}$, resulting in a new process $\hat{\mathbb{Y}}$. Similarly, we construct $\hat{\mathbb{Z}}$ by convolving with $R_{XZ}$ yielding
\begin{align}
R_{X\hat{Y}}&=R_{XY}*R_{YX},\\
R_{X\hat{Z}}&=R_{XZ}*R_{ZX},
\end{align}
which are symmetric functions. Set $\bar{\mathbf{X}}={\bf Q}^{\top}X^n$, $\bar{\mathbf{Y}}={\bf Q}^{\top}\hat{Y}^n$, $\bar{\mathbf{Z}}={\bf Q}^{\top}\hat{Z}^n$ where $\bf Q$ the sine/cosine orthogonal matrix, i.e., for $1\le k,l\le n$,
\begin{align}\label{sincos}
Q_{kl}:=\left\{\begin{array}{cc}\bigskip
                 \cos(\frac{2\pi}{n}\lfloor\frac{l}{2}\rfloor k) & \textrm{~$l$ is odd;} \\
                 \sin(\frac{2\pi}{n}\lfloor\frac{l}{2}\rfloor k) & \textrm{~$l$ is even.}
               \end{array}
\right.
\end{align}
Then the covariance matrices $\bf\Sigma_{\bar{X}}$, $\bf\Sigma_{\bar{Y}}$, $\bf\Sigma_{\bar{Z}}$, $\bf\Sigma_{\bar{X}\bar{Y}}$, $\bf\Sigma_{\bar{X}\bar{Z}}$ will be asymptotically diagonal as their dimension grows.

In summary, the original Gaussian sources are converted to sources satisfying the product assumption \eqref{assump1} and \eqref{assump2} in the spectral representation, and the correlation coefficients corresponding to frequency $\omega$ (which relates to the factor $\frac{2\pi}{n}\lfloor\frac{l}{2}\rfloor$ in \eqref{sincos}) are
\begin{align}
\rho_{XY}(\omega)&:=\frac{|S_{XY}(\omega)|}{\sqrt{S_{X}(\omega)S_{Y}(\omega)}},\label{rho}\\
\rho_{XZ}(\omega)&:=\frac{|S_{XZ}(\omega)|}{\sqrt{S_{X}(\omega)S_{Z}(\omega)}},\label{rho1}
\end{align}
where $S_{X},S_{Y},S_{Z},S_{XY},S_{XZ}$ denote the spectral densities and joint spectral densities. From (\ref{rho}), (\ref{rho1}) and Theorem~\ref{thm4},
%a result of asymptotic distribution of eigenvalues \cite[p. 63]{grenander1958toeplitz},
we can anticipate the expression in the next result. To prove it rigorously we impose a technical condition that requires all correlations and cross-correlations to be absolutely summable (that is, the corresponding spectrum functions are in the ``Wiener class'' \cite{gray2006toeplitz}). We do not believe this condition to be crucial for the validity of the result.
\begin{theorem}\label{thm5}
Suppose $\mathbb{X}$, $\mathbb{Y}$ and $\mathbb{Z}$ are Wiener class stationary Gaussian processes, and
\begin{align}
\beta(\omega):=\frac{|S_{XY}(\omega)|^2S_Z(\omega)-|S_{XZ}
(\omega)|^2S_Y(\omega)}{S_X(\omega)S_Y(\omega)S_Z(\omega)-
|S_{XY}(\omega)|^2S_Z(\omega)}
\end{align}
is well-defined, that is, excluding the $0/0$ case. Then the achievable communication and key rates are parameterized by $\mu>0$ as
\begin{align}
r&=\frac{1}{4\pi}\int_{\beta(\omega)>\mu}\log\frac{\beta(\omega)(\mu+1)}{(\beta(\omega)
+1)\mu}{\rm d}\omega,\label{e58}
\\
R&=\frac{1}{4\pi}\int_{\beta(\omega)>\mu}
\log\frac{\beta(\omega)+1}{\mu+1}{\rm d}\omega.\label{e59}
\end{align}
\end{theorem}
\begin{remark}
\begin{align}
\eta_{\mathbb{Z}}(\mathbb{X};\mathbb{Y})=\sup_{\omega\in[0,2\pi)}(\beta^+(\omega)).
\end{align}
\end{remark}
\begin{remark}
From \eqref{rho} and \eqref{rho1} we can verify that
\begin{align}
\beta(\omega)=\frac{\rho_{XY}^2(\omega)-\rho_{XZ}^2(\omega)}{1-\rho_{XY}^2(\omega)},
\end{align}
which is the counterpart of $\beta_i$ in Theorem~\ref{thm4}.
\end{remark}
The achievability proof of Theorem~\ref{thm5} is given in Sections~\ref{sec4} and \ref{sec5}, and the converse is relegated to Appendix~\ref{app_conv}.

%\subsection{Characterization of Initial Efficiency via Hypercontractivity (This part may not be included in the submission)}
%Suppose the sources are degraded, i.e., $X-Y-Z$. Define $s_p^*(X,Y|Z)$ as the minimum of $r$ such that the following holds for all simple functions $f$ and $g$.
%\begin{align}
%\mathbb{E}\log\mathbb{E}[f(X)g(Y)|Z]\le\mathbb{E}\log\|f\|_{L^{p'}(X|Z)}+
%\mathbb{E}\log\|g\|_{L^{pr}(Y|Z)}.
%\end{align}
%In the above, all expectations except the second one are with respect to $Z$. $p'$ is the dual exponent of $p$, and $L^{p'}(X|Z=z)$ denotes the space of $L^{p'}$ functions on $\mathcal{X}$ under the measure $P_{X|Z=z}$. Then,
%\begin{align}
%\inf_{p>1}s_p^*(X,Y|Z)=s^*_Z(X,Y)
%\end{align}
%where $s^*_Z(X,Y)$ is as defined in Section \ref{s3a}.
%
%For symmetric Bernoulli sources, it's elementary to show that the initial efficiency is
%\begin{align}\label{e79}
%s_Z^*(X,Y)=\frac{(1-2p_{XY})^2-(1-2p_{XZ})^2}{1-(1-2p_{XY})^2}.
%\end{align}
%Using a simple central limit theorem argument, we can derive the initial efficiency of Gaussian sources from (\ref{e79}).

\section{Achievability of One-Shot Key Generation}\label{sec4}
The single-letter expressions of \eqref{region0} or \eqref{region1} only apply to discrete memoryless sources. In order to allow memory, and in particular to prove the achievability part of Theorem \ref{thm5}, we derive a one-shot achievability result in this section. The proof relies on a stochastic encoding scheme called \emph{likelihood encoder} \cite{song}. The idea is to introduce an idealized distribution which is easier to work with, and which approximates the true distribution in total variation distance under certain rate conditions according to soft covering lemma/resolvability \cite{cuff2012distributed}.

\begin{notation}
Given $P_{XY}$, denote the information density by
\begin{align}
\imath_{X;Y}(x;y):=\log\frac{{\rm d}P_{XY}}{{\rm d}(P_X\times P_Y)}(x,y).
\end{align}
\end{notation}
\begin{theorem}\label{thm55}
Suppose the sources are distributed according to $P_{XYZ}$, the integers $M,M_1,M_2>0$, and $\bar{P}_{U|X}$ is a conditional distribution on an arbitrary alphabet $\mathcal{U}$. Then there is a scheme such that $|\mathcal{W}|=M$, $|\mathcal{K}|=M_1$, and that
\begin{align}
\mathbb{P}(\hat{K}\neq K)\le \epsilon^*,
\end{align}
\begin{align}\label{e62}
\log M_1-H(K|WZ)\le \inf_{0<\delta<e^{-1}M_1^{\frac{3}{2}}}\left\{(T^*+8\delta)\log\frac{M_1^{\frac{3}{2}}}{\delta}
\right\},
\end{align}
where $\epsilon^*$ and $T^*$ are defined in \eqref{es} and \eqref{ts}.
%where $T_1,T_2,T_3,\epsilon$ are defined in (\ref{d1}),(\ref{d2}),(\ref{d3}),(\ref{d4}).
\end{theorem}
\begin{proof}
Fix the joint distribution of the sources $P_{XYZ}$. Let $\bar{P}_{UXYZ}=\bar{P}_{U|X}P_{XYZ}$. Randomly generate a codebook
\begin{align}
\mathbf{U}\in \mathcal{U}^{M\times M_1\times M_2}
\end{align}
according to $\bar{P}_U$.
Let $P_{WKLXYZ}$ be the distribution induced by the likelihood encoder \cite{song}:
\begin{align}\label{enc}
P_{WKL|XYZ}(w,k,l|x,y,z)=\frac{1}{Z_x}\bar{P}_{X|U}(x|U(w,k,l))
\end{align}
where $Z_x$ is a normalization constant independent of $(w,k,l)$. In words, the stochastic encoder in \eqref{enc} outputs the indices $w$, $k$ and $l$ according to the likelihood of $U(w,k,l)$ passing through the ``test channel'' $\bar{P}_{X|U}$. Define
\begin{align}
Q_{XWKL}(x,w,k,l)&=\frac{1}{MM_1M_2}\bar{P}_{X|U}(x|U(w,k,l)),\\
Q_{YZ|XWKL}&=P_{YZ|X}.\label{qcond}
\end{align}
Note that $Q_{WKL}$ is an equiprobable distribution, hence by the construction of the likelihood encoder we have
\begin{align}\label{e67}
P_{WKL|X}=Q_{WKL|X}.
\end{align}

We now digress into a brief review of the \emph{total variation distance}. By definition, the total variation distance between probability measures $P$ and $Q$ on the same $\sigma$-algebra of subsets $\mathcal{F}$ of the sample space $\mathcal{X}$ is
\begin{align}
|P-Q|:=\sup_{\mathcal{A}\in\mathcal{F}}|P(\mathcal{A})-Q(\mathcal{A})|.
\end{align}
Below are some of the relevant properties of total variational distance; see for example \cite{cuff2012distributed}.
\begin{property}
\begin{enumerate}
\item Triangle inequality: if $P$, $Q$ and $S$ are distributions on the same sample space, then
\begin{align}
|P-Q|\le|P-S|+|S-Q|.
\end{align}
\item If $P_XP_{Y|X}$ and $Q_XQ_{Y|X}$ are joint distributions on $\mathcal{X}\times\mathcal{Y}$, then
\begin{align}
|P_X-Q_X|\le|P_XP_{Y|X}-Q_XQ_{Y|X}|
\end{align}
where the equality holds when $P_{Y|X}=Q_{Y|X}$.
\end{enumerate}
\end{property}

According to Theorem VII.1 in \cite{cuff2012distributed}, we have the following bounds on the total variations with respect to the codebook $\mathcal{C}$:
\begin{align}
\mathbb{E}_{\mathcal{C}}|Q_{Z|W=w}-P_Z|&\le T_1,\\
\mathbb{E}_{\mathcal{C}}|Q_{Z|W=w,K=k}-P_Z|&\le T_2,
\end{align}
for each $m,k$, where
\begin{align}
T_1&:=\inf_{\tau>0}\left\{\mathbb{P}[\imath_{U;Z}(U;Z)>\tau]+\frac{1}{2}\sqrt{\frac{2^{\tau}}{MM_2}}\right\},\label{d1}\\
T_2&:=\inf_{\tau>0}\left\{\mathbb{P}[\imath_{U;Z}(U;Z)>\tau]+\frac{1}{2}\sqrt{\frac{2^{\tau}}{M_2}}\right\},\label{d2}
\end{align}
and $\imath_{U;Z}(U;Z)$ is computed with the joint distribution $\bar{P}_{UZ}$.
By the triangle inequality,
\begin{align}
\mathbb{E}_{\mathcal{C}}|Q_{Z|W=w,K=k}-Q_{Z|W=w}|&\le T_1+T_2, \quad \forall w,k,
\end{align}
and since $Q_{WK}=Q_WQ_{K}$, we obtain
\begin{align}\label{eq1}
\mathbb{E}_{\mathcal{C}}|Q_{ZWK}-Q_{ZW}Q_{K}|
&=\mathbb{E}_{\mathcal{C}}|Q_{ZWK}-Q_{Z|W}Q_{WK}|
\\
&\le T_1+T_2.
\end{align}
From $P_{WKL|X}=Q_{WKL|X}$, we have
\begin{align}
\mathbb{E}_{\mathcal{C}}|P_{WKLX}-Q_{WKLX}|
&=\mathbb{E}_{\mathcal{C}}|P_X-Q_X|\\
&\le T_3,
\end{align}
where
\begin{align}\label{d3}
T_3:=&\inf_{\tau>0}\left\{\mathbb{P}[\imath_{U;X}(U;X)>\tau]+\frac{1}{2}\sqrt{\frac{2^{\tau}}{MM_1M_2}}\right\}
\end{align}
and $\imath_{U;X}(U;X)$ is computed with the joint distribution $\bar{P}_{UX}$.
Therefore by (\ref{qcond}),
\begin{align}
\mathbb{E}_{\mathcal{C}}|P_{KWZ}-Q_{KWZ}|
&\le \mathbb{E}_{\mathcal{C}}|P_{KWXZ}-Q_{KWXZ}|\\
&=\mathbb{E}_{\mathcal{C}}|P_{KWX}-Q_{KWX}|\\
&\le T_3,\label{eq2}
\end{align}
and
\begin{align}
\mathbb{E}_{\mathcal{C}}|P_{WZ}Q_K-Q_{WZ}Q_K|=
\mathbb{E}_{\mathcal{C}}|P_{WZ}-Q_{WZ}|\le T_3\label{eq3}.
\end{align}
Equations (\ref{eq1}), (\ref{eq2}), (\ref{eq3}) and the triangle inequality imply that
\begin{align}
\mathbb{E}_{\mathcal{C}}\int|P_{K|ZW}-Q_{K}|{\rm d}P_{ZW}
&= \mathbb{E}_{\mathcal{C}}|P_{KWZ}-P_{ZW}Q_{K}|\\
&\le T_1+T_2+2T_3.\label{eq4}
\end{align}
\begin{lemma}
For any $z,w$,
\begin{align}
&\quad D(P_{K|Z=z,W=w}||Q_{K})\nonumber
\\
&\le 2|P_{K|Z=z,W=w}-Q_{K}|\log\frac{M_1^{\frac{3}{2}}}{|P_{K|Z=z,W=w}-Q_{K}|}.
\end{align}
\end{lemma}
\begin{proof}
\begin{align}
&\quad D(P_{K|Z=z,W=w}||Q_{K})\nonumber
\\
&=-H(P_{K|Z=z,W=w})+\sum_{k=1}^{M_1}P_{K|Z=z,W=w}(k)\log\frac{1}{Q_{K}(k)}\\
&=-H(P_{K|Z=z,W=w})+H(Q_{K})\nonumber\\
&\quad+\sum_{k=1}^{M_1}(P_{K|Z=z,W=w}(k)-Q_{K}(k))\log M_1\\
&\le2|P_{K|Z=z,W=w}-Q_{K}|\log\frac{M_1}{|P_{K|Z=z,W=w}-Q_{K}|}\nonumber\\
&\quad+|P_{K|Z=z,W=w}-Q_{K}|\log M_1,
\end{align}
where the last step used the inequality in \cite{zhang2007estimating}.
\end{proof}
Thanks to the lemma, for any $0<\delta<e^{-1}M_1^{\frac{3}{2}}$ we have
\begin{align}
&I(K;Z,W)
\nonumber\\
&=D(P_{KZW}||P_{K}P_{ZW})\\
&\le D(P_{KZW}||Q_{K}P_{ZW})\\
&=\int D(P_{K|ZW}||Q_{K}){\rm d}P_{ZW}\\
&\le2\int|P_{K|ZW}-Q_{K}|\log\frac{M_1^{\frac{3}{2}}}{|P_{K|ZW}-Q_{K}|}{\rm d}P_{ZW}\\
&=2\log M_1^{\frac{3}{2}}|P_{KZW}-Q_{K}P_{ZW}|
\\
    &\quad+2\int|P_{K|ZW}-Q_{K}|\log\frac{1}{|P_{K|ZW}-Q_K|}{\rm d}P_{ZW}\\
&\le2|P_{KZW}-Q_{K}P_{ZW}|\left(\log M_1^{\frac{3}{2}}+\log\frac{1}{|P_{KZW}-Q_KP_{ZW}|}\right)\label{eq92}\\
&\le2\log\frac{M_1^{\frac{3}{2}}}{\delta}(|P_{KZW}-Q_{K}P_{ZW}|+\delta),\label{eq93}
%=&2\int_{z,w:|P_{K|ZW}-Q_{K}|\ge \delta}|P_{K|ZW}-Q_{K}|\log\frac{M_1^{\frac{3}{2}}}{|P_{K|ZW}-Q_{K}|}{\rm d}P_{ZW}\\
% &+2\int_{z,w:|P_{K|ZW}-Q_{K}|< \delta}|P_{K|ZW}-Q_{K}|\log\frac{M_1^{\frac{3}{2}}}{|P_{K|ZW}-Q_{K}|}{\rm d}P_{ZW}\\
%\le&2\log\frac{M_1^{\frac{3}{2}}}{\delta}\int|P_{K|ZW}-Q_{K}|{\rm d}P_{ZW}+2\delta\log\frac{M_1^{\frac{3}{2}}}{\delta}
\end{align}
where we used Jensen's inequality in \eqref{eq92} and $x\log\frac{\lambda}{x}\le (x+\delta)\log\frac{\lambda}{\delta}$ for all $x>0$ and $0<\delta<e^{-1}\lambda$ in \eqref{eq93}.
Averaging \eqref{eq93} over the codebook and applying (\ref{eq4}), we obtain
\begin{align}\label{e91}
\mathbb{E}_{\mathcal{C}}I(K;Z,W)\le 2\log\frac{M_1^{\frac{3}{2}}}{\delta}(T_1+T_2+2T_3+\delta).
\end{align}
Similarly from \eqref{eq2} we have $\mathbb{E}_{\mathcal{C}}|P_{K}-Q_{K}|\le T_3$, hence
\begin{align}\label{e92}
\mathbb{E}_{\mathcal{C}}[\log K-H(K)]=\mathbb{E}_{\mathcal{C}}D(P_{K}||Q_{K})\le2\log\frac{M_1^{\frac{3}{2}}}{\delta}(T_3+\delta).
\end{align}
Thus for the security constraint, we have
\begin{align}\label{e93}
\mathbb{E}_{\mathcal{C}}[\log K-H(K|WZ)]\le 2\log\frac{M_1^{\frac{3}{2}}}{\delta}(T_1+T_2+3T_3+2\delta),
\end{align}
which follows from \eqref{e91} and \eqref{e92}.

For the key agreement constraint, choose a good channel decoder $P_{\hat{K}|WY}$, and let
\begin{align}
P_{WKLXYZ\hat{W}}&=P_{WKLXYZ}P_{\hat{K}|WY},\\
Q_{WKLXYZ\hat{W}}&=Q_{WKLXYZ}P_{\hat{K}|WY}.
\end{align}
Then using a single-shot version of Shannon's achievability bound \cite{shannon1957certain} for discrete memoryless channels, the error probability of the channel decoder can be bounded as $\mathbb{E}_{\mathcal{C}}\mathbb{P}_{Q}(\hat{K}\neq K)
\le \epsilon$ where we have defined
\begin{align}\label{d4}
\epsilon:=\inf_{\gamma>0}\left\{\mathbb{P}[\imath_{U;Y}(U;Y)\le\log(M_1M_2-1)+\gamma]+\exp(-\gamma)\right\},
\end{align}
and $\imath_{U;Y}(U;Y)$ is computed with the joint distribution $\bar{P}_{UY}$.
Then, the probability of decoding $K$ erroneously under the true distribution is bounded as
\begin{align}
\mathbb{E}_{\mathcal{C}}\mathbb{P}_{P}(\hat{K}\neq K)
&\le \mathbb{E}_{\mathcal{C}}\mathbb{P}_{Q}(\hat{K}\neq K)+|P_{X}-Q_{X}|\label{e100}
\\
&\le T_3+\epsilon,
\end{align}
where $\mathbb{P}_{P}$ and $\mathbb{P}_{Q}$ denote the probabilities under the distributions $P_{XYWK}$ and $Q_{XYWK}$, respectively. In \eqref{e100} we used $P_{K\hat{K}WY|X}=Q_{K\hat{K}WY|X}$, which follows from $P_{WY|X}=Q_{WY|X}$ in \eqref{e67}, and that $K$ and $\hat{K}$ are functions of $X$ and $(W,Y)$, respectively.
By Markov's inequality,
\begin{align}
\mathbb{P}_{\mathcal{C}}[\mathbb{P}_{P}(\hat{K}\neq K)> 2(T_3+\epsilon)]<\frac{1}{2}.
\end{align}
Similarly from \eqref{e93},
\begin{align}
&\mathbb{P}_{\mathcal{C}}\left[\log K-H(K|WZ)> 4(T_1+T_2+3T_3+2\delta)\log\frac{M_1^{\frac{3}{2}}}{\delta}\right]
\nonumber
\\
&<\frac{1}{2}.
\end{align}
Hence there exists a codebook which satisfies the properties in Theorem~\ref{thm55} where
\begin{align}
\epsilon^*&:=2(T_3+\epsilon),\label{es}
\\
T^*&:=4(T_1+T_2+3T_3);\label{ts}
\end{align}
and $T_1$, $T_2$, $T_3$, $\epsilon$ are as in \eqref{d1}, \eqref{d2}, \eqref{d3} and \eqref{d4}.
\end{proof}

\section{Approximation of Gaussian Processes and Achievability of Theorem \ref{thm5}}\label{sec5}
In this section we apply Theorem~\ref{thm55} to stationary Gaussian processes to finish the achievability part of Theorem~\ref{thm5}. The derivation is essentially based on the asymptotic distribution of the eigenvalues of Toeplitz matrices, a brief review of which is given in Appendix~\ref{appI}.

We now introduce notations for Toeplitz matrices and circulant matrices. Given a continuous function $f$ on $[0,2\pi)$, define for $k=0,1,\dots,n-1$,
\begin{align}
t_k&:=\frac{1}{2\pi}\int_0^{2\pi}f(\omega)e^{ik\omega}{\rm d}\omega,\\
c^{(n)}_k&:=\sum_{m=-\infty}^{\infty}t_{-k+mn}.\label{def1}
\end{align}
Note that from (\ref{def1}), an equivalent way of defining $c^{(n)}_k$ is
\begin{align}\label{def2}
c^{(n)}_k:=\frac{1}{n}\sum_{j=0}^{n-1}f(2\pi j/n)e^{2\pi ijk/n}.
\end{align}
If $\{t_k\}$ has fast decay, then $\{c^{(n)}_k\}$ approximates $\{t_k\}$ for large $n$. The advantage of $\{c^{(n)}_k\}$ over $\{t_k\}$ is that the former is a periodic sequence.
For $0\le i,j\le n-1$, define
\begin{align}
[\mathbf{T}_n(f)]_{i,j}&:=t_{i-j},\label{e109}\\
[\mathbf{C}_n(f)]_{i,j}&:=c_{i-j}.\label{e110}
\end{align}
Then it is clear that \eqref{e110} is a circulant matrix.

Using the above notations, the covariance matrix of the vector $(X^n,Y^n,Z^n)$ which are samples from $(\mathbb{X},\mathbb{Y},\mathbb{Z})$ can be expressed as
\begin{align}
\mathbf{T}_n:=&\left(
     \begin{array}{ccc}
       \mathbf{T}_n(S_X) & \mathbf{T}_n(S_{XY}) & \mathbf{T}_n(S_{XZ}) \\
       \mathbf{T}_n(S_{YX}) & \mathbf{T}_n(S_Y) & \mathbf{T}_n(S_{YZ}) \\
       \mathbf{T}_n(S_{ZX}) & \mathbf{T}_n(S_{ZY}) & \mathbf{T}_n(S_Z) \\
     \end{array}
   \right),\\
\end{align}
Now define a positive-semidefinite matrix composed of circulant blocks
\begin{align}\label{ecirc}
   \mathbf{C}_n:=&\left(
     \begin{array}{ccc}
       \mathbf{C}_n(S_X) & \mathbf{C}_n(S_{XY}) & \mathbf{C}_n(S_{XZ}) \\
       \mathbf{C}_n(S_{YX}) & \mathbf{C}_n(S_Y) & \mathbf{C}_n(S_{YZ}) \\
       \mathbf{C}_n(S_{ZX}) & \mathbf{C}_n(S_{ZY}) & \mathbf{C}_n(S_Z) \\
     \end{array}
   \right).
\end{align}
We assume that all the spectrums belong to the Wiener class. Then from Fact~\ref{fact6} in Appendix~\ref{appI} we have
\begin{align}\label{e116}
\mathbf{T}_n\sim \mathbf{C}_n
\end{align}
since the corresponding blocks in $\mathbf{T}_n$ and $\mathbf{C}_n$ are asymptotically equivalent.
We shall use $\mathbf{C}_n$ as a proxy for $\mb{T}_n$ in the subsequent analysis. Let $(\tilde{X}^n,\tilde{Y}^n,\tilde{Z}^n)$ be a zero mean Gaussian vector with covariance matrix $\mathbf{C}_n$. Suppose $\mb{Q}$ is the sin/cosine orthogonal matrix (see \eqref{sincos}). Define
\begin{align}
\mb{\hat{X}}&=\mb{Q}^{\top}\mb{\tilde{X}},\label{e123}\\
\mb{\hat{Y}}&=\mb{Q}^{\top}\mathbf{C}_n\left(\frac{S_{XY}}{|S_{XY}|}\right)\mb{\tilde{Y}},\label{e124}\\
\mb{\hat{Z}}&=\mb{Q}^{\top}\mathbf{C}_n\left(\frac{S_{XZ}}{|S_{XZ}|}\right)\mb{\tilde{Z}}.\label{e125}
\end{align}
Here $\frac{S_{XY}(\omega)}{|S_{XY}(\omega)|}$ can be arbitrarily set to $1$ if $S_{XY}(\omega)=0$. This ensures that $\mb{C}_n(\frac{S_{XY}(\omega)}{|S_{XY}(\omega)|})$ is an invertible, and in particular, unitary matrix. Note that the simplified discussion in \ref{sec3c} corresponds to replacing $\mathbf{C}_n\left(\frac{S_{XY}}{|S_{XY}|}\right)$ in \eqref{e124} with $\mathbf{C}_n\left(S_{XY}\right)$, which may be singular.
One can verify that $(\mb{\hat{X}},\mb{\hat{Y}},\mb{\hat{Z}})$ has the product structure of \eqref{assump1} and \eqref{assump2}. Next we shall specify an auxiliary distribution $P_{\mb{\hat{U}}|\mb{\hat{X}}}$. We first design the correlation coefficients $\rho_{UX}:[0,2\pi)\to[0,1]$ as
\begin{align}\label{e54}
&\rho_{UX}(\omega)=
\nonumber\\
&\left\{\begin{array}{cc}
                                                  \left(\frac{(1+\mu)\rho^2_{XY}(\omega)-\rho^2_{XZ}(\omega)-\mu}
                                                  {\rho^2_{XY}(\omega)-(1+\mu)\rho^2_{XZ}(\omega)
                                                  +\mu\rho^2_{XY}(\omega)\rho^2_{XZ}(\omega)}
                                                  \right)^{\frac{1}{2}}
                                                   & \beta(\omega)>\mu \\
                                                  0 & \textrm{otherwise} \\
                                                \end{array}
\right.
\end{align}
%for $\omega\in\frac{2\pi}{n}\mathbb{Z}$, which can be calculated from the optimization (Thm 2 in ISIT paper).
for $\omega\in[0,2\pi)$, where $\rho^2_{XY}(\omega)$ and $\rho^2_{XZ}(\omega)$ are as in \eqref{rho} and \eqref{rho1}. The definition \eqref{e54} ensures that $\rho_{UX}$ satisfies
\begin{align}\label{e118}
\log\frac{\beta(\omega)(\mu+1)}{(\beta(\omega)+1)\mu}
=\log\frac{1}{1-\rho^2_{UX}(\omega)}-\log\frac{1}{1-\rho^2_{UX}(\omega)\rho^2_{XY}(\omega)}
\end{align}
and
\begin{align}\label{e119}
&\log\frac{\beta(\omega)+1}{\mu+1}
=
\nonumber\\
&\log\frac{1}{1-\rho^2_{UX}(\omega)\rho^2_{XY}(\omega)}-\log\frac{1}{1-\rho^2_{UX}(\omega)\rho^2_{XZ}(\omega)}.
\end{align}
The intuition for $\rho_{UX}$ is as follows: suppose $\mathbb{U}$ is a Gaussian process jointly stationary with $\mathbb{X}$ and $\mathbb{U}-\mathbb{X}-(\mathbb{Y},\mathbb{Z})$ such that $\frac{|S_{UX}(\omega)|}{\sqrt{S_X(\omega)S_U(\omega)}}=\rho_{UX}(\omega)$. Then from \eqref{e118}, \eqref{e119} and Theorem \ref{thm5} we can verify a counterpart of the rate region \eqref{region1} for stationary processes:
\begin{align}
I(\mathbb{U};\mathbb{X})-I(\mathbb{U};\mathbb{Y})&=r,\\
I(\mathbb{U};\mathbb{Y})-I(\mathbb{U};\mathbb{Z})&=R,
\end{align}
where $I(\mathbb{U};\mathbb{X}):=\lim_{n\to\infty}\frac{1}{n}I(U^n;X^n)$ stands for the mutual information rate between $\mathbb{U}$ and $\mathbb{X}$.
Now, $P_{\hat{U}_i|\hat{X}_i}$ can be defined by requiring that $\hat{U}_i$ is zero mean jointly Gaussian with $\hat{X}_i$ satisfying
\begin{align}\label{eq124}
\rho_{\hat{U}_i\hat{X}_i}
=\rho_{UX}\left(\frac{2\pi i}{n}\right),~i=1,2,\dots,n
\end{align}
The scaling of $\hat{U}_i$ doesn't matter and can be chosen arbitrarily. We set $P_{\bf \hat{U}|\hat{X}}=\prod_{i=1}^n P_{\hat{U}_i|\hat{X}_i}$. Notice that this and (\ref{e123})-(\ref{e125}) have defined a channel $P_{\bf\hat{U}|\tilde{X}}$. Also beware that $\hat{X}_i$ and $\tilde{X}_i$ (and $\bar{X}_i$ to be defined later) depend implicitly on $n$, though $X_i$ does not. Below, $\rho_{\hat{U}_i\hat{X}_i}$ will be denoted by $\rho^{(n)}_i$ for simplicity.

Now for $i=0,\dots,n-1$, define the random variables
\begin{align}\label{eq123}
\eta_i^{(n)}=\imath_{\hat{X}_i;\hat{U}_i}(\hat{X}_i;\hat{U}_i).
\end{align}
%For any $B\in\mathbb{R}$, Markov inequality (or Chernoff bound) gives
%\begin{align}\label{e53}
%\frac{1}{n}\ln\mathbb{P}\left(\frac{1}{n}\sum_{i=1}^n\eta_i^{(n)}\ge B\right)
%\le \frac{1}{n}\sum_{i=1}^n\ln\mathbb{E}e^{t\eta^{(n)}_i}-tB
%\end{align}
%for any $t>0$. If the right hand side above is negative, then $\mathbb{P}\left(\frac{1}{n}\sum_{i=1}^n\eta_i\ge B\right)$ have exponential decay. We shall use the following lemma:
The following lemma will be useful later when applying Chernoff bound:
\begin{lemma}\label{lemma3}
Fix any $0<\delta<\frac{1}{2}$. For any $\epsilon>0$, there exists $t>0$ such that
\begin{align}\label{el1}
t\mathbb{E}\eta\le\ln\mathbb{E}e^{t\eta}\le (1+\epsilon)t\mathbb{E}\eta+\epsilon t
\end{align}
for all $\rho\in[\delta-1,1-\delta]$, where $\eta:=\imath_{U;X}(U;X)$, in which $U,X$ are jointly Gaussian with correlation coefficient $\rho$.
\end{lemma}
\begin{proof}
See Appendix \ref{app_lem3}.
\end{proof}

Now return to the proof of Theorem \ref{thm5}. Define
\begin{align}\label{eq152}
\delta:=1-\sup_{0\le \omega<2\pi }\rho_{UX}(\omega).
\end{align}

From the assumption of Theorem~\ref{cor1}, we know that $S_X(\omega)$, $S_Z(\omega)$ and $S_Z(\omega)$ do not vanish for any $\omega\in[0,2\pi)$, since otherwise $\beta(\omega)$ will be a fraction of the type $\frac{0}{0}$ for some $\omega$.
This in turn implies that
\begin{align}\label{e153}
\min_{0\le \omega<2\pi }S_X(\omega)>0,\quad\min_{0\le \omega<2\pi }S_Y(\omega)>0,\quad \min_{0\le \omega<2\pi }S_Z(\omega)>0;
\end{align}
since $S_X(\omega)$, $S_Y(\omega)$ and $S_Z(\omega)$ are continuous functions on the compact set $[0,2\pi)$. We shall make an additional assumption that
\begin{align}\label{e131}
\quad\sup_{0\le \omega<2\pi }\rho^2_{XY}(\omega)<1.
\end{align}
Fortunately, the proof does not lose any generality due to the assumptions of \eqref{e131}:
\begin{lemma}\label{assumptionlem}
If Theorem~\ref{thm5} holds for sources satisfying \eqref{e131}, then it must also hold without those assumptions.
\end{lemma}
\begin{proof}
Assume that Theorem~\ref{thm5} is proved under the assumptions \eqref{e131}. For general source $(\mathbb{X},\mathbb{Y},\mathbb{Z})$ and $\lambda\in[0,1)$, we can degrade $\mathbb{Y}$ by $\mathbb{Y}^{\lambda}:=\mathbb{Y}+\lambda\mathbb{N}$, where $\mathbb{N}$ is a stationary white Gaussian processes such that $\mathbb{N}$ and $(\mathbb{X},\mathbb{Y},\mathbb{Z})$ are independent.
Let $\beta^{\lambda}(\omega)$ be as defined in Theorem~\ref{thm5} but for the new source $(\mathbb{X},\mathbb{Y}^{\lambda},\mathbb{Z})$, and define
\begin{align}
r^{\lambda}&:=\frac{1}{4\pi}\int_{\beta^{\lambda}(\omega)>\mu}\log\frac{\beta^{\lambda}(\omega)(\mu+1)}
{(\beta^{\lambda}(\omega)
+1)\mu}{\rm d}\omega,\label{eq134}
\\
R^{\lambda}&:=\frac{1}{4\pi}\int_{\beta^{\lambda}(\omega)>\mu}
\log\frac{\beta^{\lambda}(\omega)+1}{\mu+1}{\rm d}\omega.\label{eq135}
\end{align}
It's easy to check that $\beta^{\lambda}(\omega)\uparrow\beta(\omega)$ as $\lambda\downarrow0$ for each $\omega\in[0,2\pi)$. Then by monotone convergence theorem we have $r^{\lambda}\uparrow r$ and $R^{\lambda}\uparrow R$ as $\lambda\downarrow0$, where $r$ and $R$ are as in \eqref{e58} and \eqref{e59}. However for each $\lambda>0$ the condition \eqref{e131} holds.
By our assumption we can prove $(r^{\lambda},R^{\lambda})\in\mathcal{R}(\mathbb{X},\mathbb{Y}^{\lambda},\mathbb{Z})$, and the Markov chain $\mathbb{Y}^{\lambda}-\mathbb{Y}-(\mathbb{X},\mathbb{Z})$ implies $\mathcal{R}(\mathbb{X},\mathbb{Y}^{\lambda},\mathbb{Z})\subseteq\mathcal{R}(\mathbb{X},\mathbb{Y},\mathbb{Z})$; hence we also have $(r^{\lambda},R^{\lambda})\in\mathcal{R}(\mathbb{X},\mathbb{Y},\mathbb{Z})$. Then by the closure property of the achievable region we know $(R,r)\in\mathcal{R}(\mathbb{X,Y,Z})$.
\end{proof}

Assume that \eqref{e153} and \eqref{e131} are true. If $\beta(\omega):=\frac{\rho^2_{XY}(\omega)-\rho^2_{XZ}(\omega)}{1-\rho^2_{XY}(\omega)}>\mu$ then from \eqref{e54},
\begin{align}
&\inf_{0\le\omega<2\pi}\{1-\rho^2_{UX}(\omega)\}
\nonumber
\\
&=\inf_{0\le\omega<2\pi}\frac{\mu(1-\rho^2_{XY}(\omega)(1-\rho^2_{XZ}(\omega)))}
{\rho^2_{XY}(\omega)-(1+\mu)\rho^2_{XZ}(\omega)+\mu\rho^2_{XY}(\omega)\rho^2_{XZ}(\omega)}\label{eq156}
\\
&\ge\inf_{0\le\omega<2\pi}\mu\frac{1-\rho^2_{XY}(\omega)}{\rho^2_{XY}(\omega)}\label{eq157}
\\
&>0.\label{e132}
\end{align}
where \eqref{eq157} used the monotonically increasing property of the rational function on the right hand side of \eqref{eq156} in $\rho^2_{XZ}(\omega)\in[0,(1+\mu)\rho^2_{XY}(\omega)-\mu)$. This means that $\delta>0$ in \eqref{eq152}, which will be essential to applying Lemma~\ref{lemma3}.
%Then from the above lemma, for any $\epsilon>0$ there exist $t>0$ such that for all $n$,
%\begin{align}\label{e55}
%\frac{1}{n}\sum_{i=1}^n\ln\mathbb{E}e^{t\eta^{(n)}_i}-tB
%&\le t(1+\epsilon)\frac{1}{n}\sum_{i=1}^n\mathbb{E}\eta^{(n)}_i+\epsilon t-tB.
%\end{align}

For Wiener class Gaussian processes, the spectral function is continuous. Hence from \eqref{eq124}, \eqref{eq123} and the definition of Riemann integral we have
\begin{align}
\frac{1}{n}\sum_{i=1}^n\mathbb{E}\eta^{(n)}_i
&\to\frac{1}{4\pi}\int\log\left(\frac{1}{1-\rho^2_{UX}(\omega)}\right){\rm d}\omega
\\
&=I(\mathbb{U};\mathbb{X}).\label{e56}
\end{align}
Now fix $B>I(\mathbb{U};\mathbb{X})$. Define $P_{\mb{U}|\mb{X}}:=P_{\mb{\hat{U}}|\mb{\tilde{X}}}$.
According to Corollary \ref{cor2}, there exist non-degenerate linear transforms on $\mb{U}$ and $\mb{X}$ to obtain $\mb{\bar{U}}$ and $\mb{\bar{X}}$ such that $P_{\mb{\bar{U}}\mb{\bar{X}}}=\prod_{i=1}^n P_{\bar{U}_i\bar{X}_i}$. Let $\bar{\rho}^{(n)}_i,~i=1,\dots,n$ be the correlation coefficients between $\bar{U}_i$ and $\bar{X}_i$. From the proof of Lemma~\ref{lem3} one can verify that $(\bar{\rho}^{(n)}_i)^2$, $i=1,\dots,n$ are eigenvalues of $\mb{I}-{\bf \Sigma_X}^{-\frac{1}{2}}{\bf \Sigma_{X|U}}{\bf \Sigma_X}^{-\frac{1}{2}}$, and $(\rho^{(n)}_i)^2$, $i=1,\dots,n$ are eigenvalues of $\mb{I}-{\bf \Sigma_{\tilde{X}}}^{-\frac{1}{2}}{\bf \Sigma_{\tilde{X}|\hat{U}}}{\bf \Sigma_{\tilde{X}}}^{-\frac{1}{2}}$. However these two matrices are asymptotically equivalent, and their largest eigenvalues are uniformly upper bounded away from one, which follows immediately from Fact~\ref{fact4} and the following result.
\begin{lemma}\label{lem4}
Under the assumptions \eqref{e153} and \eqref{e131}, we have
\begin{description}
  \item[(a)] \begin{align}
\bf \Sigma_{X}\sim\Sigma_{\tilde{X}},\quad
\Sigma_{Y}\sim\Sigma_{\tilde{Y}},\quad
\Sigma_{Z}\sim\Sigma_{\tilde{Z}}.
\end{align}
Moreover, the smallest eigenvalues of these matrices are uniformly bounded (meaning that the bound is independent of $n$) away from zero, and their largest eigenvalues are also uniformly upper bounded.

  \item[(b)] \begin{align}
\bf \Sigma_{X|U}\sim\Sigma_{\tilde{X}|\hat{U}},\quad
\Sigma_{Y|U}\sim\Sigma_{\tilde{Y}|\hat{U}},\quad
\Sigma_{Z|U}\sim\Sigma_{\tilde{Z}|\hat{U}}.
\end{align}
Moreover, the smallest eigenvalues of these matrices are uniformly bounded away from zero.
\end{description}
\end{lemma}
\begin{proof}
See Appendix \ref{appf}.
\end{proof}
Therefore $\{(\bar{\rho}_i^{(n)})^2\}$ is asymptotically equally distributed as $\{(\rho_i^{(n)})^2\}$ on $[0,1-\delta_0)$ for some $\delta_0>0$ according to Fact~\ref{fact5}. It follows that for any continuous function $F$ on $[0,1-\delta_0)$,
\begin{align}\label{e157}
\lim_{n\to\infty}\frac{1}{n}\sum_iF((\bar{\rho}_i^{(n)})^2)=\lim_{n\to\infty}\frac{1}{n}\sum_i F((\rho_i^{(n)})^2).
\end{align}
Define $\bar{\eta}_i^{(n)}=\imath_{\bar{X}_i;\bar{U}_i}(\bar{X}_i;\bar{U}_i)$. Then fixing $\epsilon<\frac{B-I(\mathbb{U};\mathbb{X})}{3+I(\mathbb{U};\mathbb{X})}$, there exists $t>0$ such that for all $n$,
\begin{align}
\frac{1}{n}\ln\mathbb{P}\left(\imath_{\mb{X};\mb{U}}(\mb{X};\mb{U})\ge nB\right)
&= \frac{1}{n}\ln\mathbb{P}\left(\frac{1}{n}\sum_{i=1}^n\bar{\eta}_i^{(n)}\ge B\right)\label{e57}
\\
&\le \frac{1}{n}\sum_{i=1}^n\ln\mathbb{E}e^{t\bar{\eta}^{(n)}_i}-tB\label{e159}
\\
&\le t(1+\epsilon)\frac{1}{n}\sum_{i=1}^n\mathbb{E}\bar{\eta}^{(n)}_i+\epsilon t-tB\label{e160}
\end{align}
where \eqref{e159} is from Markov's inequality (or the Chernoff bound) and \eqref{e160} uses Lemma \ref{lemma3} and the fact that $|\bar{\rho}_i^{(n)}|<\sqrt{1-\delta_0}$. Now let $F\colon x\mapsto \frac{1}{2}\log(\frac{1}{1-x})$. From \eqref{e157} and \eqref{e56}, there exists $n_0>0$ such that for $n>n_0$,
\begin{align}\label{e161}
\frac{1}{n}\sum_{i=1}^n\mathbb{E}\bar{\eta}^{(n)}_i<I(\mathbb{U};\mathbb{X})+\frac{\epsilon}{1+\epsilon}.
\end{align}
Then \eqref{e160} and \eqref{e161} imply that for $n>n_0$,
\begin{align}
\frac{1}{n}\ln\mathbb{P}\left(\imath_{\mb{X};\mb{U}}(\mb{X};\mb{U})\ge nB\right)
&<t[(1+\epsilon)I(\mathbb{U};\mathbb{X})+2\epsilon-B]
\\
&<-t\epsilon.\label{e163}
\end{align}

To finish the achievability proof, we need to show that the bounds in Theorem \ref{thm55} converge to zero for rate pairs in the interior of $\mathcal{R}(\mathbb{X,Y,Z})$. An inspection of the bounds in Theorem \ref{thm55} reveals that it suffices to show (as $n\to\infty$)
\begin{enumerate}
  \item $\mathbb{P}(\imath_{\mb{U};\mb{X}}(\mb{U};\mb{X})>nB)$ converges to $0$ exponentially fast;
  \item $\mathbb{P}(\imath_{\mb{U};\mb{Y}}(\mb{U};\mb{Y})<nC)$ converges to $0$;
  \item $\mathbb{P}(\imath_{\mb{U};\mb{Z}}(\mb{U};\mb{Z})>nD)$ converges to $0$ exponentially fast,
\end{enumerate}
for $P_{\bf UXYZ}:=P_{\bf \hat{U}|\tilde{X}}P_{\bf XYZ}$, and any $B>I(\mathbb{U};\mathbb{X})$, $C<I(\mathbb{U};\mathbb{Y})$ and $D>I(\mathbb{U};\mathbb{Z})$. Speed of converge is imposed in 1) and 3), so that upon choosing $\delta$ to be exponentially decreasing in $n$, the term
\begin{align}
T^*+8\delta=4(T_1+T_2+3T_3+2\delta)
\end{align}
in \eqref{e62} is also exponentially decreasing in $n$, thus annihilating the term $\log\frac{M_1^{\frac{3}{2}}}{\delta}$ in \eqref{e62}, which grows linearly in $n$. From \eqref{e163} we see the validity of property~1).

The proof of 3) follows the same steps as that of 1). Similar to \eqref{e56}, we have
\begin{align}
\frac{1}{n}I({\bf\hat{U};\tilde{Z}})
&=\frac{1}{n}I({\bf\hat{U};\hat{Z}})
\\
&\to\frac{1}{4\pi}\int\log\left(\frac{1}{1-\rho^2_{XU}(\omega)\rho^2_{XZ}(\omega)}\right){\rm d}\omega
\\
&=I(\mathbb{U};\mathbb{Z}).\label{e149}
\end{align}
And as in \eqref{e57}-\eqref{e160}, fixing $\epsilon<\frac{D-I(\mathbb{U};\mathbb{Z})}{3+I(\mathbb{U};\mathbb{Z})}$ there exists $t>0$ so that we can upper bound
\begin{align}
\frac{1}{n}\ln\mathbb{P}\left(\imath_{\mb{Z};\mb{U}}(\mb{Z};\mb{U})\ge nD\right)
\le t(1+\epsilon)\frac{1}{n}I({\bf U;Z})+\epsilon t-tD.\label{eq150}
\end{align}
Then \eqref{e149} and \eqref{eq150} will imply 3) once
\begin{align}\label{e151}
\lim_{n\to\infty}\frac{1}{n}[I({\bf U;Z})-I({\bf\hat{U};\tilde{Z}})]=0
\end{align}
is established. Now suppose $\bf U\to \overline{\underline{U}}$ and $\bf Z\to \overline{\underline{Z}}$ are the diagonalizing linear transforms in Lemma~\ref{lem3}. Then it suffices to show that $\{\rho_{\overline{\underline{U}}_i;\overline{\underline{Z}}_i}^2\}_{i=1}^n$ and $\{\rho_{\hat{U}_i;\hat{Z}_i}^2\}_{i=1}^n$ are asymptotically equally distributed on $[0,1-\delta_0]$.
Indeed, we first note that $\max_{1\le i\le n}|\rho_{\overline{\underline{U}}_i;\overline{\underline{Z}}_i}|$ is the maximal correlation coefficient between $\bf U$ and $\bf Z$,
and $\max_{1\le i\le n}\rho_{\bar{U}_i;\bar{X}_i}$ is the maximal correlation coefficient between $\bf U$ and $\bf X$,
hence $\max_{1\le i\le n}|\rho_{\overline{\underline{U}}_i;\overline{\underline{Z}}_i}|\le\max_{1\le i\le n}\rho_{\bar{U}_i;\bar{X}_i}\le\sqrt{1-\delta_0}$ due to the Markov chain $\bf U-X-Z$. By a similar argument we also have $\max_{1\le i\le n}|\rho_{\hat{U}_i;\hat{Z}_i}|\le \sqrt{1-\delta_0}$. Hence we have shown that $\rho_{\hat{U}_i;\hat{Z}_i}^2$ and $\rho_{\overline{\underline{U}}_i;\overline{\underline{Z}}_i}^2$ are bounded in $[0,1-\delta_0]$. To show their asymptotic equidistribution, it remains to prove that
\begin{align}\label{eqq152}
{\bf I}-{\bf \Sigma_Z}^{-\frac{1}{2}}{\bf \Sigma_{Z|U}}{\bf \Sigma_Z}^{-\frac{1}{2}}\sim
{\bf I}-{\bf \Sigma_{\tilde{Z}}}^{-\frac{1}{2}}{\bf \Sigma_{\tilde{Z}|\hat{U}}}{\bf \Sigma_{\tilde{Z}}}^{-\frac{1}{2}}
\end{align}
which follows immediately from Lemma~\ref{lem4} and Fact~\ref{fact4}.

The proof of 2) is simpler: without an requirement on the speed of convergence, we can just use a coarse upper bounded via Chebyshev's inequality:
\begin{align}\label{eqq156}
\mathbb{P}(\imath_{\mb{U};\mb{Y}}(\mb{U};\mb{Y})<nC)\le
\frac{{\rm Var}(\imath_{\mb{U};\mb{Y}}(\mb{U};\mb{Y}))}{n^2(\frac{1}{n}I(\mb{U};\mb{Y})-C)^2}
\end{align}
The roles of $\bf Y$ and $\bf Z$ are identical to the counterparts of \eqref{e149} and \eqref{e151} hold, so we have
\begin{align}\label{eqq157}
\lim_{n\to\infty}\frac{1}{n}I({\bf U;Y})=I(\mathbb{U};\mathbb{Y}).
\end{align}
Suppose $\bf U\to\underline{{U}}$ and $\bf Z\to\underline{{Y}}$ are the diagonalizing linear transforms in Lemma~\ref{lem3}. Then as before $\max_{1\le i\le n}\rho_{\underline{{U}}_i;\underline{{Y}}_i}^2\le 1-\delta_0$ which is uniformly upper bounded for all $n$. Hence there exists a uniform upper bound ${\rm Var}(\imath_{\underline{{U}}_i;\underline{{Y}}_i}
(\underline{{U}}_i;\underline{{Y}}_i))<V$ for some $V>0$ independent of $n$. Then ${\rm Var}(\imath_{\mb{U};\mb{Y}}(\mb{U};\mb{Y}))\le nV$, and so condition 2) is true by virtue of \eqref{eqq156} and \eqref{eqq157}. The achievability proof for Theorem~\ref{thm5} is completed.

\begin{remark}
Although the assumption that $\beta(\omega)$ is well defined for each $\omega\in[0,2\pi)$ in Theorem~\ref{thm5} is fairly reasonable, it is still possible that $\beta(\omega)$ is \emph{not} defined for a set of frequencies of measure zero yet the Lebesgue integrals in \eqref{e58} and \eqref{e59} still make sense. In such a case, we no longer have the convenient conditions in \eqref{e153}. However, if only the first two conditions in \eqref{e153} are unfulfilled and $\min_{0\le \omega<2\pi }S_Z(\omega)>0$ remains true, we can still prove Theorem~\ref{thm5} by showing the achievability for some degraded $\mathbb{X}$ and $\mathbb{Y}$ first and then applying the closure property of the achievable region, which is similar to the argument in Lemma~\ref{assumptionlem}. Nonetheless, our proof cannot be easily extended to the case where $\min_{0\le \omega<2\pi }S_Z(\omega)=0$, since degrading the eavesdropper's observation can only augment the achievable region.
\end{remark}

\section{Discussion}\label{sec6}
As remarked earlier, Theorem \ref{thm3} is analogous to a rate distortion theorem for product sources under additive distortion measure; in fact one can show a similar result for channel capacity with additive cost constraints.
Related phenomena in information theory also include the additivity of channel capacity (without input constraints) and Wyner's common information \cite{wyner1975common}. In those cases, the achievable rate region of the product source/channel is the Minkowski sum of the achievable region of the factor sources/channels.
The evidence points to the principle that rate splitting is optimal for product resources asymptotically in most information theoretic problems admitting single-letter solutions.\footnote{Exceptions to this principle do exist, for example the key generation with an omniscient helper problem \cite{liu2015key}, the mismatched broadcast channel with a common message \cite[Remark~9.6]{el2011network}, and lossy compression with mismatched side-information \cite{watanabe2013}. \label{ft7}} Indeed, the algebraic manipulations in the converse proofs usually rely only on the independence of $\{X_t\}$, and do not require them to be identically distributed. Hence the main element in proving such a result about rate splitting (e.g. Lemma \ref{lem1} in the appendix) is usually related to the converse proof of the corresponding coding theorem.
However, there are a number of examples where the \emph{achievable} regions fail to satisfy such an additive property (c.f.~a relay broadcast channel discussed in \cite[Remark~17]{liang2007rate}), although the \emph{exact} region is not known.
Moreover, this rule also fails quite often for coding problems of combinatorial nature. For example, the additivity of zero error capacity was a famous conjecture \cite{shannon1959coding}\cite{lovasz1979shannon} which has now been disproved \cite{alon1998shannon}.

It is also interesting to consider the constant
\begin{align}
s_*(X;Y):=\inf_{U-X-Y,I(U;X)\neq0}\frac{I(U;Y)}{I(U;X)}=\inf_{Q_X\neq P_X}\frac{D(Q_Y||P_Y)}{D(Q_X||P_X)}
\end{align}
where $Q_X\to P_{Y|X}\to Q_Y$. Interestingly, $s_*(X;Y)$ does not tensorize, and in fact it usually vanishes exponentially in $L$ for i.i.d. $(X_i,Y_i)_{i=1}^L$. Indeed if $H(X|Y)>0$, we can choose $R\in(I(X;Y),H(X))$. Set $P_{Y^L|X^L}=\prod_{i=1}^L P_{Y_i|X_i}$ and $P_{X^L}=\prod_{i=1}^L P_{X_i}$. By resolvability/soft covering lemma and its strong converse \cite{wyner1975common}\cite{han1993approximation}\cite{cuff2012distributed}, we can choose a $\lfloor2^{nR}\rfloor$-type\footnote{In \cite{han1993approximation} a probability distribution $P$ is called $M$-type if $P(a)M$ is an integer for each $a\in\mathcal{A}$.} distribution $Q_{X^L}$ and set $Q_{X^L}\to P_{Y^L|X^L}\to Q_{Y^L}$ such that $D(Q_{Y^L}||P_{Y^L})$ converges to zero exponentially as $n\to\infty$ whereas $D(Q_{X^L}||P_{X^L})$ is bounded away from zero, from which the exponential decay of $\frac{D(Q_{Y^L}||P_{Y^L})}{D(Q_{X^L}||P_{X^L})}$ follows. This implies, among other things, that no information theoretic problem can have a single-letter solution of the form $[I(U;Y),\infty)\times [0,I(U;X)]$.

\section{Acknowledgments}
We are pleased to acknowledge Sanket Satpathy for suggesting the Minkowski sum interpretation in Theorem~\ref{thm3}, and Shun Watanabe for pointing out the last two examples in footnote~\ref{ft7}. This work was supported by NSF under Grants CCF-1350595, CCF-1116013, CCF-1319299, CCF-1319304, and the Air Force Office of Scientific Research under Grant FA9550-15-1-0180, FA9550-12-1-0196.

\appendices
\section{A Key Observation for Product Sources}\label{app1}
The following observation is central to the proof of both tensorization property of $s^*_Z(X;Y)$ and the optimality of rate splitting in Theorem \ref{thm3}. It thus manifests how the two problems are inherently connected.
\begin{lemma}\label{lem1}
Suppose that $\{(X_i,Y_i,Z_i)\}_{i=1}^L$ possess the product structure of \ref{assump1} and \eqref{assump2}, and $(U,V)$ are r.v.'s such that $(U,V)-X^L-(Y^L,Z^L)$. Then there exist $U^L$ and $V^L$ such that $(U_i,V_i)-X_i-(Y_i,Z_i)$ for $i=1,\dots,L$ and
\begin{align}
I(U,V;X^L)-I(U,V;Y^L)
&\ge\sum_{i=1}^L [I(U_i,V_i;X_i)-I(U_i,V_i;Y_i)],\label{e43}\\
I(V;Y^L|U)-I(V;Z^L|U)
&=\sum_{i=1}^L [I(V_i;Y_i|U_i)-I(V_i;Z_i|U_i)].\label{e1}
\end{align}
\end{lemma}
\begin{proof}
Suppose we are given the additional condition that $Y^L-X^L-Z^L$ form a Markov chain, then \eqref{assump1} and \eqref{assump2} will imply $P_{X^LY^LZ^L}=\prod_{i=1}^L P_{X_i,Y_i,Z_i}$ which will facilitate the proof. Now in general $Y^L-X^L-Z^L$ may not be true; but notice that the expressions in \eqref{e43} and \eqref{e1} depend only on the marginal distributions of $Y^L$ and $Z^L$ given $X^L$, rather than how they are correlated given $X^L$. Hence we can convert the source distribution to a new one where $Y^L-X^L-Z^L$ while the conditional marginal distributions of $(X^L,Y^L)$ and $(X^L,Z^L)$ remain the same.

To carry out the above procedure, choose $\bar{Z}^L$ such that
\begin{align}
&P_{UVX^LY^L\bar{Z}^L}(u,v,x^L,y^L,z^L)\nonumber\\
&=P_{UVX^L}(u,v,x^L)P_{Y^L|X^L}(y^L|x^L)P_{Z^L|X^L}(z^L|x^L).
\end{align}
Define $\bar{U}_i=(Y^{i-1},\bar{Z}_{i+1}^L,U)$ and $V_i=V$. Then $(\bar{U}_i,V_i)-X_i-(Y_i,\bar{Z}_i)$ for each $i$. Moreover
\begin{align}\label{eq184}
I(V;Y^L|U)-I(V;\bar{Z}^L|U)=&\sum_{i=1}^L [I(V_i;Y_i|\bar{U}_i)-I(V_i;\bar{Z}_i|\bar{U}_i)]
\end{align}
holds, which is a standard identity in multiuser information theory (see for example \cite[Lemma 4.1]{ahlswede1993common}),

Next, observe that
\begin{align}
&I(U,V;X^L)-I(U,V;Y^L)\nonumber\\
&=\sum_{i=1}^L[I(U,V;X_i|X_{i+1}^L,Y^{i-1})-I(U,V;Y_i|X_{i+1}^LY^{i-1})]\label{st1}\\
&=\sum_{i=1}^L[I(U,V,X_{i+1}^L,Y^{i-1};X_i)-I(U,V,X_{i+1}^LY^{i-1};Y_i)]\label{st2}\\
&=\sum_{i=1}^L[I(U,V,X_{i+1}^L,Y^{i-1},\bar{Z}^L_{i+1};X_i)\nonumber\\
&\quad\quad\quad-I(U,V,X_{i+1}^LY^{i-1},\bar{Z}^L_{i+1};Y_i)]\label{st3}\\
&=\sum_{i=1}^L[I(U,V,Y^{i-1},\bar{Z}^L_{i+1};X_i)-I(U,V,Y^{i-1},\bar{Z}^L_{i+1};Y_i)]\nonumber\\
&+\sum_{i=1}^L[I(X_{i+1}^L;X_i|U,V,Y^{i-1},\bar{Z}^L_{i+1})\nonumber\\
&\quad\quad\quad-I(X_{i+1}^L;Y_i|U,V,Y^{i-1},\bar{Z}^L_{i+1})]\\
&\ge\sum_{i=1}^L[I(U,V,Y^{i-1},\bar{Z}^L_{i+1};X_i)-
I(U,V,Y^{i-1},\bar{Z}^L_{i+1};Y_i)]\label{st4}\\
&=\sum_{i=1}^L[I(\bar{U}_i,V_i;X_i)-I(\bar{U}_i,V_i;Y_i)],\label{eq190}
\end{align}
where (\ref{st1}) is again an application of \cite[Lemma 4.1]{ahlswede1993common}, and (\ref{st2}) is from the independence $(X_i,Y_i)\perp(X_{i+1}^L,Y^{i-1})$. Equality (\ref{st3}) follows from the Markov condition $\bar{Z}_{i+1}^L-(U,V,X_{i+1}^L,Y^{i-1})-(X_i,Y_i)$.
%which in turn follows from the Markov chain $\bar{Z}_{i+1}^L-X_{i+1}^L-(U,V,X_i,Y^i)$.
Inequality (\ref{st4}) is because of $I(X_{i+1}^L;X_i|U,V,Y^{i-1},\bar{Z}^L_{i+1})=I(X_{i+1}^L;X_i,Y_i|U,V,Y^{i-1},\bar{Z}^L_{i+1})$, due to the Markov condition $X_{i+1}^L-(U,V,Y^{i-1},\bar{Z}^L_{i+1},X_i)-Y_i$.
%, which in turn comes from $(X_{i+1}^n,U,V,Y^{i-1},\bar{Z}^n_{i+1})-X_i-Y_i$.

Finally, for each $i$ let $U_i$ be a r.v. such that
\begin{align}
P_{U_i|V_iX_iY_iZ_i}(u_i|v_i,x_i,y_i,z_i)=P_{\bar{U}_i|V_iX_i}(u_i|v_ix_i).
\end{align}
Then $(U_i,V_i,X_i,Z_i)$ and $(\bar{U}_i,V_i,X_i,\bar{Z}_i)$ have the same distribution, hence
\begin{align}
I(V_i;Z_i|U_i)&=I(V_i;\bar{Z}_i|\bar{U}_i).\label{e185}
\end{align}
By the same token, we also have
\begin{align}
I(V_i;Y_i|U_i)&=I(V_i;Y_i|\bar{U}_i),
\\
I(U_i,V_i;Y_i)&=I(\bar{U}_i,V_i;Y_i),
\\
I(U_i,V_i;X_i)&=I(\bar{U}_i,V_i;X_i),
\\
I(V;Z^L|U)&=I(V;\bar{Z}^L|U).\label{e184}
\end{align}
Therefore we see that \eqref{eq184}, \eqref{eq190} imply the desired result, once we make the substitutions with \eqref{e185}-\eqref{e184}.
\end{proof}

%\begin{conjecture}
%Suppose $V-X^n-Z^n$ and $P_{X^n}=\prod_{i=1}^nP_{X_i}$.
%Define $\bar{Z}^n$ via
%\begin{align}
%P_{VX^n\bar{Z}^n}(v,x^n,z^n)
%=P_{VX^n}(v,x^n)\prod_{i=1}^n P_{Z_i|X_i}(z_i|x_i).
%\end{align}
%Claim: $I(V;Z^n)\ge I(V;\bar{Z}^n)$.
%\end{conjecture}
%\begin{remark}
%When $V=X^n$, this can be proved simply using chain rule.
%\end{remark}

In the case where $Z^L$ does not exist, the tensorization property of $s^*(X;Y)$ and Theorem \ref{thm3} can also be proved using the following result.
\begin{lemma}\label{lem5}
Suppose that $\{(X_i,Y_i)\}_{i=1}^L$ possess the product structure of \eqref{assump1} and \eqref{assump2}, and $U$ is a r.v. such that $U-X^L-Y^L$. Then there exist $U^L$ such that $U_i-X_i-Y_i$ for $i=1,\dots,L$ and
\begin{align}
I(U;X^L)&=\sum_{i=1}^L I(U_i;X_i),\label{e179}
\\
I(U;Y^L)&\le\sum_{i=1}^L I(U_i;Y_i).\label{e180}
\end{align}
\end{lemma}
\begin{proof}
By induction, it suffices to prove the case of $L=2$. Let $U_1:=U$ and $U_2:=(U,X_1)$. We have:
\begin{align}
I(U;X^2)&=I(U;X_1)+I(U;X_2|X_1)\nonumber\\
&=I(U;X_1)+[I(U;X_2|X_1)+I(X_1;X_2)]\nonumber\\
&=I(U;X_1)+I(U,X_1;X_2)\nonumber,
\end{align}
\begin{align}
I(U;Y^2)&=I(U;Y_1)+I(U;Y_2|Y_1)\nonumber\\
&=I(U;Y_1)+[I(U;Y_2|Y_1)+I(Y_1;Y_2)]\nonumber\\
&=I(U;Y_1)+I(U,Y_1;Y_2)\nonumber\\
&\le I(U;Y_1)+I(U,X_1,Y_1;Y_2)\nonumber\\
&=I(U;Y_1)+I(U,X_1;Y_2).
\end{align}
where the last equality is from the Markov chain $Y_1-(U,X_1)-Y_2$.
\end{proof}
Note that setting $Z^L$ in Lemma \ref{lem1} to be a constant will imply the existence of $U^L$ satisfying $U_i-X_i-Y_i$ and the inequalities
\begin{align}
I(U;X^L)&\ge\sum_{i=1}^L I(U_i;X_i),
\\
I(U;Y^L)&=\sum_{i=1}^L I(U_i;Y_i),
\end{align}
which are different from (\ref{e179}) and (\ref{e180}). Hence Lemma \ref{lem5} is not a special case of Lemma \ref{lem1}.

\section{Proof of Theorem\ref{thm1}}\label{app2}
%The proof of 1) is based only data processing inequality and is omitted here.
\begin{enumerate}
\item From the data processing inequality the denominator in (\ref{ss}) is nonnegative, and $s^*_Z(X;Y)\le 1$. If there exists $U$ such that
$I(U;X)-I(U;Y)>0$, we can choose $V$ independent of $U,X,Y$ so that the numerator vanishes whereas the denominator is positive, which shows that $s^*_Z(X;Y)\ge0$. Otherwise if $I(U;X)-I(U;Y)=0$ for all $U$, the numerator will always be nonnegative:
\begin{align}
&\quad I(V;Y|U)-I(V;Z|U)\nonumber\\
&=I(U,V;Y)-I(U,V;Z)-I(U;Y)+I(U;Z)\nonumber\\
&=I(U,V;X)-I(U,V;Z)-I(U;X)+I(U;Z).
\end{align}
Hence $s^*_Z(X;Y)\ge0$ always holds.

Of course, from an operational viewpoint $0\le s^*_Z(X;Y)\le 1$ must be true because of Part 3) as well.

\item We only show that
\begin{align}
s^*_Z(X^L;Y^L)\le \max_{1\le i\le n}s^*_{Z_i}(X_i;Y_i)
\end{align}
since the other direction is trivial. For any $U,V$ such that $(U,V)-X^L-(Y^L,Z^L)$ and both
\begin{align}
I(U,V;X^L)-I(U,V;Y^L)>0
\end{align}
and
\begin{align}
I(V;Y^L|U)-I(V;Z^L|U)>0,
\end{align}
let $U^L,V^L$ be as in Lemma~\ref{lem1} in the appendix. That is, $U^L,V^L$ are such that $(U_i,V_i)-X_i-(Y_i,Z_i)$ for each $i$ and both
\begin{align}
I(U,V;X^L)-I(U,V;Y^L)\ge\sum_{i=1}^L [I(U_i,V_i;X_i)-I(U_i,V_i;Y_i)]
\end{align}
and
\begin{align}
I(V;Y^L|U)-I(V;Z^L|U)=\sum_{i=1}^L [I(V_i;Y_i|U_i)-I(V_i;Z_i|U_i)]
\end{align}
hold. Then
\begin{align}
&\frac{I(V;Y^L|U)-I(V;Z^L|U)}{I(U,V;X^L)-I(U,V;Y^L)}
\nonumber\\
&\le \frac{\sum_{i=1}^L [I(V_i;Y_i|U_i)-I(V_i;Z_i|U_i)]}{\sum_{i=1}^L [I(U_i,V_i;X_i)-I(U_i,V_i;Y_i)]}\label{leftmost}\\
&\le \max_{i\in \mathcal{I}}\frac{I(V_i;Y_i|U_i)-I(V_i;Z_i|U_i)}
{I(U_i,V_i;X_i)-I(U_i,V_i;Y_i)}\nonumber\\
&\le \max_{1\le i\le L}\sup_{U_i,V_i}\frac{I(V_i;Y_i|U_i)-I(V_i;Z_i|U_i)}
{I(U_i,V_i;X_i)-I(U_i,V_i;Y_i)}
\end{align}
%\begin{figure*}[b]%\usepackage{dblfloatfix}
%\hrulefill
%\begin{align}
%%\hline\nonumber\\
%&\frac{I(V;Y|U)-I(V;Z|U)}
%{I(V;X|U)-I(V;Z|U)+I(U;X)-I(U;Y)}\nonumber\\
%=&\frac{\int[I(V;Y|U=u)-I(V;Z|U=u)]{\rm d}P_U(u)}{\int[I(V;X|U=u)-I(V;Z|U=u)
%+D(P_{X|U=u}||P_X)-D(P_{Y|U=u}||P_Y)]{\rm d}P_U(u)}\label{equ18}\\
%\le&\sup_{u}\frac{I(V;Y|U=u)-I(V;Z|U=u)}{I(V;X|U=u)-I(V;Z|U=u)
%+D(P_{X|U=u}||P_X)-D(P_{Y|U=u}||P_Y)}\\
%\le&\sup_{Q_{VX}}\frac{I(\bar{V};\bar{Y})-I(\bar{V};\bar{Z})}
%{I(\bar{V};\bar{X})-I(\bar{V};\bar{Z})+D(Q_X||P_X)-D(Q_Y||P_Y)}.\label{equ20}
%\end{align}
%\end{figure*}
where $\mathcal{I}$ is the set of indices such that $I(U_i,V_i;X_i)-I(U_i,V_i;Y_i)\neq 0$, and the suprema are over all $U_i,V_i$ such that $(U_i,V_i)-X_i-(Y_i,Z_i)$ and $I(U_i,V_i;X_i)-I(U_i,V_i;Y_i)\neq 0$. Supremizing with respect to $U,V$ on the left hand side of (\ref{leftmost}) shows the tensorization property of $\frac{s^*_Z(X;Y)}{1-s^*_Z(X;Y)}$, which is equivalent to the tensorization property of $s^*_Z(X;Y)$.

\item One direction of the inequality is trivial: if $U$ and $V$ are such that $I(V;Y|U)-I(V;Z|U)\ge 0$, we have (see \eqref{equ18}-\eqref{equ20})
\begin{figure*}[b]%\usepackage{dblfloatfix}
\hrulefill
\begin{align}
%\hline\nonumber\\
&\frac{I(V;Y|U)-I(V;Z|U)}
{I(V;X|U)-I(V;Z|U)+I(U;X)-I(U;Y)}\nonumber\\
&=\frac{\int[I(V;Y|U=u)-I(V;Z|U=u)]{\rm d}P_U(u)}{\int[I(V;X|U=u)-I(V;Z|U=u)
+D(P_{X|U=u}||P_X)-D(P_{Y|U=u}||P_Y)]{\rm d}P_U(u)}\label{equ18}\\
&\le\sup_{u}\frac{I(V;Y|U=u)-I(V;Z|U=u)}{I(V;X|U=u)-I(V;Z|U=u)
+D(P_{X|U=u}||P_X)-D(P_{Y|U=u}||P_Y)}\\
&\le\sup_{Q_{VX}}\frac{I(\bar{V};\bar{Y})-I(\bar{V};\bar{Z})}
{I(\bar{V};\bar{X})-I(\bar{V};\bar{Z})+D(Q_X||P_X)-D(Q_Y||P_Y)}.\label{equ20}
\end{align}
\end{figure*}
For the other direction, to construct a distribution on $(U,V,X,Y,Z)$ from $Q$, we use a binary $U$ biased heavily toward zero. When $U=1$, the distribution is as specified by $Q$. When $U=0$, $V$ is independent of $(X,Y,Z)$, and the marginal distribution on $X$ balanced slightly to counteract $Q$, so that on average the distribution on $X$ is the source distribution. Even though this distribution is only rarely behaving according to $Q$ (i.e. only when $U=1$, which has low probability), we will see that the quantity of interest only depends on $Q$. Formally, for any $Q_{VX}$, consider
\begin{align}
Q^{(1)}_{VX}&:=Q_{VX},\\
Q^{(0)}_{VX}&:=P_V\cdot\frac{P_{X}-\alpha Q^{(1)}_{X}}{1-\alpha},\label{eq23}\\
P^{\alpha}_{XUV}(x,u,v)
&:=(1-\alpha)Q^{(0)}_{VX}(v,x)1_{u=0}
\nonumber\\
&\quad+\alpha Q^{(1)}_{VX}(v,x)1_{u=1},
\end{align}
where $P_V$ is an arbitrary probability distribution on $\mathcal{V}$. Then clearly $P^{\alpha}_X=P_X$ for each $0<\alpha<1$. Finally, define
\begin{align}
P^{\alpha}_{XYZUV}:=P^{\alpha}_{XUV}P_{YZ|X}.
\end{align}
In \eqref{e18} we have implicitly assumed that $D(Q_X||P_X)$ is well defined and so the support of $Q_X$ is a subset of the support of $P_X$.
Thus (\ref{eq23}) is a well-defined distribution for $\alpha>0$ small enough. Then, we can verify that $P^{\alpha}_{XYZ}=P_{XYZ}$ and the Markov chain $(U,V)-X-(Y,Z)$ with respect to $P^{\alpha}$. Next observe that (see \eqref{e188}-\eqref{e187})
\begin{figure*}[b]%\usepackage{dblfloatfix}
\hrulefill
\begin{align}
&\frac{I(V;Y|U)-I(V;Z|U)}
{I(V;X|U)-I(V;Z|U)+I(U;X)-I(U;Y)}\nonumber\\
&=\frac{\alpha[I(V;Y|U=1)-I(V;Z|U=1)]}
{\alpha[I(V;X|U=1)-I(V;Z|U=1)]+I(U;X)-I(U;Y)}\label{e188}\\
&=\frac{\alpha[I(\bar{V};\bar{Y})-I(\bar{V};\bar{Z})]}
{\alpha[I(\bar{V};\bar{X})-I(\bar{V};\bar{Z})]+\alpha[D(Q_X||P_X)-D(Q_Y||P_Y)]
        +(1-\alpha)[D(Q_{X|U=0}||P_X)-D(Q_{Y|U=0}||P_Y)]}\\
&=\frac{I(\bar{V};\bar{Y})-I(\bar{V};\bar{Z})}
{I(\bar{V};\bar{X})-I(\bar{V};\bar{Z})+D(Q_X||P_X)-D(Q_Y||P_Y)+o(1)}\label{e187}
\end{align}
\end{figure*}
as $\alpha\downarrow0$, where $(U,V,X,Y,Z)$ has the joint distribution $P_{UVXYZ}:=P^{\alpha}_{UVXYZ}$, and the distribution of $(\bar{V},\bar{X},\bar{Y},\bar{Z})$ is as in \eqref{e19}. Equation \eqref{e188} is from the independence between $V$ and $(X,Y,Z)$ under $U=0$. To justify \eqref{e187}, recall the property of relative entropy that if $P_{\lambda}:=\lambda P_1+(1-\lambda)P_0$ is a distribution for sufficiently small $\lambda>0$, then $D(P_{\lambda}||P_0)=o(\lambda)$. This smoothness condition implies that
\begin{align}
D(P_{X|U=0}||P_X)&=o(\alpha),
\\
D(P_{Y|U=0}||P_Y)&=o(\alpha).
\end{align}
Therefore \eqref{e187} is true, and the $\ge$ part of (\ref{e18}) holds.
%The proofs of (\ref{eq19}) and (\ref{eq20}) are omitted since they follow the same structure.

\item In the case of degraded sources $X-Y-Z$, we can write (see \eqref{e203}-\eqref{e206})
\begin{figure*}[b]%\usepackage{dblfloatfix}
\hrulefill
\begin{align}
&\frac{I(\bar{V};\bar{Y})-I(\bar{V};\bar{Z})}
{I(\bar{V};\bar{X})-I(\bar{V};\bar{Z})+D(Q_X||P_X)-D(Q_Y||P_Y)}
\label{e203}
\\
&=\frac{\int[D(P_{\bar{Y}|\bar{V}=v}||P_Y)-D(P_{\bar{Z}|\bar{V}=v}||P_Z)]{\rm d}P_{\bar{V}}(v)-D(P_{\bar{Y}}||P_Y)+D(P_{\bar{Z}}||P_Z)}
{\int[D(P_{\bar{X}|\bar{V}=v}||P_X)-D(P_{\bar{Z}|\bar{V}=v}||P_Z)]{\rm d}P_{\bar{V}}(v)-D(P_{\bar{Y}}||P_Y)+D(P_{\bar{Z}}||P_Z)}
\\
&\le\frac{\int[D(P_{\bar{Y}|\bar{V}=v}||P_Y)-D(P_{\bar{Z}|\bar{V}=v}||P_Z)]{\rm d}P_{\bar{V}}(v)}
{\int[D(P_{\bar{X}|\bar{V}=v}||P_X)-D(P_{\bar{Z}|\bar{V}=v}||P_Z)]{\rm d}P_{\bar{V}}(v)}
\\
&\le\sup_{Q_X}\frac{D(Q_Y||P_Y)-D(Q_Z||P_Z)}{D(Q_X||P_X)-D(Q_Z||P_Z)},
\label{e206}
\end{align}
\end{figure*}
where the first inequality is from $-D(P_{\bar{Y}}||P_Y)+D(P_{\bar{Z}}||P_Z)\le 0$, and the second inequality used the fact that $D(P_{\bar{Y}|\bar{V}=v}||P_Y)-D(P_{\bar{Z}|\bar{V}=v}||P_Z)\ge0$. This establishes the ``$\le$'' part of (\ref{eq19}). Conversely, for any $Q_X$, define
\begin{align}
Q^{(1)}_{X}&=Q_{X},\\
Q^{(0)}_{X}&=\frac{P_{X}-\alpha Q^{(1)}_{X}}{1-\alpha},\\
Q^{\alpha}_{VX}(v,x)&=\alpha Q^{(1)}_{VX}(v,x)1_{v=1}+
(1-\alpha)Q^{(0)}_{VX}(v,x)1_{v=0}.
\end{align}
Let $P_{\bar{V}\bar{X}\bar{Y}\bar{Z}}=Q^{\alpha}_{VX}P_{YZ|X}$. Notice that $P_{\bar{X}\bar{Y}\bar{Z}}=P_{XYZ}$. Then
\begin{align}
&\lim_{\alpha\downarrow 0}\frac{I(\bar{V};\bar{Y})-I(\bar{V};\bar{Z})}
{I(\bar{V};\bar{X})-I(\bar{V};\bar{Z})+D(Q^{\alpha}_X||P_X)-D(Q^{\alpha}_Y||P_Y)}
\\
&=\lim_{\alpha\downarrow 0}\frac{I(\bar{V};\bar{Y})-I(\bar{V};\bar{Z})}
{I(\bar{V};\bar{X})-I(\bar{V};\bar{Z})}
\\
&=\frac{D(Q_Y||P_Y)-D(Q_Z||P_Z)}{D(Q_X||P_X)-D(Q_Z||P_Z)}
\end{align}
This implies that
\begin{align}
&\sup_{R_{VX}}\frac{I(\bar{V};\bar{Y})-I(\bar{V};\bar{Z})}
{I(\bar{V};\bar{X})-I(\bar{V};\bar{Z})+D(R_X||P_X)-D(R_Y||P_Y)}
\nonumber\\
&\ge \frac{D(Q_Y||P_Y)-D(Q_Z||P_Z)}{D(Q_X||P_X)-D(Q_Z||P_Z)}\label{e34}
\end{align}
where $P_{\bar{V}\bar{X}\bar{Y}\bar{Z}(v,x,y,z)}=R_{VX}(v,x)P_{YZ|X}(y,z|x)$.
The proof of (\ref{eq19}) is complete since $Q_X$ in the right side of (\ref{e34}) is arbitrary.

\item If the sources are of the form $X=(X',Z)$, $Y=(Y',Z)$, we have
\begin{align}
&\frac{D(Q_Y||P_Y)-D(Q_Z||P_Z)}{D(Q_X||P_X)-D(Q_Z||P_Z)}
\nonumber\\
&=\frac{\int D(Q_{Y'|Z=z}||P_{Y'|Z=z}){\rm d}Q_Z(z)}{\int D(Q_{X'|Z=z}||P_{X'|Z=z}){\rm d}Q_Z(z)}\nonumber\\
&\le \esssup_{z\in\mathcal{Z}}\sup_{Q_{X'}}\frac{D(Q_{Y'}||P_{Y'|Z=z})}{D(Q_{X'}||P_{X'|Z=z})},\label{eq40}
\end{align}
where in the last supremum $Q_{X'Y'}=Q_{X'}P_{Y'|X'Z=z}$. Conversely for any $z_0\in\mathcal{Z}$ and $Q_{X'}$ in (\ref{eq40}), define
\begin{align}
\tilde{Q}_X(x)
&=\tilde{Q}_{X'Z}(x',z)
\nonumber
\\
&=P_Z(z_0)Q_{X'}(x')1_{z=z_0}
\\
&\quad+P_Z(z)P_{X'|Z=z}(x')1_{z\neq z_0}.
\end{align}
Then
\begin{align}
\frac{D(\tilde{Q}_Y||P_Y)-D(\tilde{Q}_Z||P_Z)}{D(\tilde{Q}_X||P_X)
-D(\tilde{Q}_Z||P_Z)}=\frac{D(Q_{Y'}||P_{Y'|Z=z_0})}{D(Q_{X'}||P_{X'|Z=z_0})}.
\end{align}
This establishes (\ref{eq20}).

\item When $Z$ is constant, we recover $s^*(X;Y)=\sup_{Q_U\neq P_U}\frac{I(U;Y)}{I(U;X)}$ by either (\ref{eq19}) or (\ref{eq20}).
\end{enumerate}

\section{Proof of Lemma \ref{lem3}}\label{appc}
With the invertible linear transform $\tilde{\mb{X}}:=\mb{\Sigma}^{-1/2}_{\mb{X}}\mb{X}$, we have
\begin{align}
\mb{\Sigma}_{\tilde{\mb{X}}}=\left(
                   \begin{array}{cc}
                     \mb{I}_{r_x} & \mb{0} \\
                     \mb{0} & \mb{0} \\
                   \end{array}
                 \right),
\end{align}
where $r_x={\rm rank}(\mb{\Sigma}_{\mb{X}})$. Similar structures are also present in $\mb{\Sigma}_{\tilde{\mb{Y}}}$ and $\mb{\Sigma}_{\tilde{\mb{Z}}}$. By positive-semidefiniteness of the covariance matrix, we have the form
\begin{align}
\mb{\Sigma}_{\tilde{\mb{X}},\tilde{\mb{Y}}}&=
\left(
  \begin{array}{cc}
    \mb{A}_{x,y} & \mb{0} \\
    \mb{0} & \mb{0} \\
  \end{array}
\right),\nonumber\\
\mb{\Sigma}_{\tilde{\mb{X}},\tilde{\mb{Z}}}&=
\left(
  \begin{array}{cc}
    \mb{A}_{x,z} & \mb{0} \\
    \mb{0} & \mb{0} \\
  \end{array}
\right),
\end{align}
where $\mb{A}_{x,y}$ and $\mb{A}_{x,z}$ are $r_x\times r_y$ and $r_x\times r_z$ matrices, respectively. However, we also have
\begin{align}
\mb{\Sigma}_{\tilde{\mb{X}},\tilde{\mb{Y}}}&=
\mb{\Sigma}^{-1/2}_{\mb{X}}\mb{\Sigma}_{\mb{X},\mb{Y}}\mb{\Sigma}^{-1/2}_{\mb{Y}},\\
\mb{\Sigma}_{\tilde{\mb{X}},\tilde{\mb{Z}}}&=
\mb{\Sigma}^{-1/2}_{\mb{X}}\mb{\Sigma}_{\mb{X},\mb{Z}}\mb{\Sigma}^{-1/2}_{\mb{Z}},
\end{align}
Hence if $\mb{G}$ and $\mb{H}$ as defined in \eqref{defm} and \eqref{defn} commute, then so do $\mb{A}_{x,y}\mb{A}_{y,x}$ and $\mb{A}_{x,z}\mb{A}_{z,x}$.
%in which case they share the same eigenspaces. This implies that we can find orthogonal matrices $\mb{Q}_x$, $\mb{Q}_y$, $\mb{Q}_z$ such that
%$
%\mb{Q}_y\mb{A}_{y,x}\mb{Q}^{\top}_x,
%$
%and
%$
%\mb{Q}_z\mb{A}_{z,x}\mb{Q}^{\top}_x,
%$
%are diagonal. To see this, define recursively
%\begin{align}
%\mb{v}_{i}\in&\arg\max_{\mb{x}\perp{\rm span}(\mb{v}_1\dots \mb{v}_{i-1}),\|\mb{x}\|_2=1}\|\mb{A}_{y,x}\mb{x}\|_2
%\\
%=&\arg\max_{\mb{x}\perp{\rm span}(\mb{v}_1\dots \mb{v}_{i-1}),\|\mb{x}\|_2=1}\|\mb{A}_{z,x}\mb{x}\|_2,
%\quad i=1\dots r_x.
%\end{align}
%and let $\mb{Q}_x=\left(\mb{v}_1\dots \mb{v}_{r_x}\right)^{\top}$.
Since commuting matrices are simultaneously diagonalizable \cite{horn2012matrix}, that is, there exists an orthogonal matrix $\mb{Q}_x$ such that $\mb{Q}_x\mb{A}_{x,y}\mb{A}_{y,x}\mb{Q}^{\top}_x$ and $\mb{Q}_x\mb{A}_{x,z}\mb{A}_{z,x}\mb{Q}^{\top}_x$ are diagonal. This in turn implies the existence of $\mb{Q}_y$ and $\mb{Q}_z$ such that
$
\mb{Q}_y\mb{A}_{y,x}\mb{Q}^{\top}_x
$
and
$
\mb{Q}_z\mb{A}_{z,x}\mb{Q}^{\top}_x
$ are diagonal.
Therefore, after the transforms
\begin{align}
\mb{X}&\mapsto \bar{\mb{X}}:=\left(
             \begin{array}{cc}
               \mb{Q}_x & \mb{0} \\
               \mb{0} & \mb{I}_{n-r_x} \\
             \end{array}
           \right)
           {\mb\Sigma}_{\mb{X}}^{-1/2}
\mb{X},\\
\mb{Y}&\mapsto \bar{\mb{Y}}:=\left(
             \begin{array}{cc}
               \mb{Q}_y & \mb{0} \\
               \mb{0} & \mb{I}_{n-r_y} \\
             \end{array}
           \right)
           {\mb\Sigma}_{\mb{Y}}^{-1/2}
\mb{Y},\\
\mb{Z}&\mapsto \bar{\mb{Z}}:=\left(
             \begin{array}{cc}
               \mb{Q}_z & \mb{0} \\
               \mb{0} & \mb{I}_{n-r_z} \\
             \end{array}
           \right)
           {\mb\Sigma}_{\mb{Z}}^{-1/2}
\mb{Z},\\
\end{align}
the matrices $\bf \Sigma_{\bar{X}},\Sigma_{\bar{Y}},\Sigma_{\bar{Z}},\Sigma_{\bar{X}\bar{Y}}$ and $\bf \Sigma_{\bar{X}\bar{Z}}$ are diagonal.

Conversely, if the asserted linear transforms exist, then there must exist orthogonal matrices $\mb{Q}_y$ and $\mb{Q}_z$ such that
$
\mb{Q}_y\mb{A}_{y,x}\mb{Q}^{\top}_x
$
and
$
\mb{Q}_z\mb{A}_{z,x}\mb{Q}^{\top}_x
$ are diagonal. Hence $\mb{A}_{x,y}\mb{A}_{y,x}$ and $\mb{A}_{x,z}\mb{A}_{z,x}$ commute, and so do $\mb{G}$ and $\mb{H}$.

\section{Connection between Gaussian and Bernoulli Sources in Example~\ref{ex2}}\label{ex2_app}
Suppose $P_{\bf UVW}=\prod_{i=1}^L P_{U_iV_iW_i}$, and $U_i,V_i$ and $W_i$ are symmetric Bernoulli random variable such that
\begin{align}
1-2\mathbb{P}[U_i\neq V_i]=\rho_{XY},\quad 1-2\mathbb{P}[U_i\neq W_i]=\rho_{XZ}.
\end{align}
Define
\begin{align}
\bar{X}_L:=\frac{1}{L}\sum_{i=1}^L U_i,\quad \bar{Y}_L:=\frac{1}{L}\sum_{i=1}^L V_i,\quad \bar{Z}_L:=\frac{1}{L}\sum_{i=1}^L W_i,
\end{align}
where the additions are on $\mathbb{R}$.
Assuming without loss of generality that $X,Y$ and $Z$ have unit variances, then by central limit theorem $P_{\bar{X}_L\bar{Y}_L}$ and $P_{\bar{X}_L\bar{Z}_L}$ converge to $P_{XY}$ and $P_{XZ}$ as $L\to\infty$, hence we expect (without a formal proof here) that $\eta_Z(X;Y)=\lim_{L\to \infty}\eta_{\bar{Z}_L}(\bar{X}_L;\bar{Y}_L)$. Observe that
\begin{align}
\eta_{\bar{Z}_L}(\bar{X}_L;\bar{Y}_L)&=\eta_{\bf W}(\bar{X}_L;\bar{Y}_L)\label{e26}
\\
&\le\eta_{\bf W}({\bf U;V})\label{e27}
\\
&=\eta_{W_1}({U_1;V_1})\label{e28}
\end{align}
where \eqref{e26} is because $\bar{Z}_L$ is a sufficient statistic of $\mb{W}$ for $(\bar{X}_L,\bar{Y}_L)$;
\eqref{e27} is because processing $\mb{U}$ and $\mb{V}$ reduces key capacity;  and \eqref{e28} uses the tensorization property \eqref{eq16}. Then from \eqref{e23} we see $\eta_Z(X;Y)\le \frac{\rho_{XY}^2-\rho_{XZ}^2}{1-\rho_{XY}^2}$. Note that this central limit argument is similar to a celebrated proof of Gaussian hypercontractivity using Boolean hypercontractivity due to Leonard Gross \cite{gross1975logarithmic}, which illustrates the interesting connection between Gaussian and symmetric Bernoulli distributions.

\section{Proof of Theorem~\ref{cor1}}\label{app_cor1}
Recall the following facts from linear algebra (see for example \cite{tao}):
\begin{fact}\label{fact2}
If $\bf A$ and $\bf B$ are matrices of the same dimension, then $\bf AB^{\top}$ and $\bf B^{\top}A$ have the same nonzero eigenvalues.
\end{fact}
%Indeed Sylvester determinant identity (c.f. \cite{tao}) states that $|\bf I+AB^{\top}|=|\bf I+B^{\top}A|$, from which it follows that $|\lambda {\bf I-AB^{\top}}|=\lambda^m|{\bf I}-\frac{1}{\lambda}{\bf AB}^{\top}|=\lambda^m|{\bf I}-\frac{1}{\lambda}{\bf B^{\top}A}|=\lambda^{m-n}|\bf \lambda I-B^{\top}A|$ for $\lambda\neq0$, where we have assumed $m\times n$ to be the size of $\mb{A}$.
\begin{fact}\label{fact3}
If $\bf A$ is a square matrix, then
\begin{align}
|{\bf I}+\epsilon {\bf A}|=\mb{I}+\epsilon\,{\rm tr}(\mb{A})+O(\epsilon^2).
\end{align}
\end{fact}
Now we are in the position of proving Theorem \ref{cor1}. Let $s:=\lambda_{\max}((\mb{G}-\mb{H})(\mb{I}-\mb{H})^{-1})$. We first show that $s^*_{\mb{Z}}(\mb{X};\mb{Y})\ge s^+$. Since $s^*_{\mb{Z}}(\mb{X};\mb{Y})$ is nonnegative we only need to focus on the case of $s\ge0$. By restricting $Q_{V\mb{X}}$ in \eqref{e18} to have the marginal distribution $P_{\mb{X}}$ on $\mathcal{X}$, we find
\begin{align}\label{e215}
s^*_{\mb{Z}}(\mb{X};\mb{Y})\ge \sup_{P_{V|\mb{X}}}\frac{I(V;\mb{Y})-I(V;\mb{Z})}{I(V;\mb{X})-I(V;\mb{Z})}.
\end{align}
We remark that using Fact \ref{fact1} one can actually show that \eqref{e215} holds with equality, although we shall not use the ``$\le$'' direction.

Let
\begin{align}
(\mb{I}-\mb{H})^{-\frac{1}{2}}(\mb{G}-\mb{H})^{\frac{1}{2}}(\mb{I}-\mb{H})^{-\frac{1}{2}}
=\mb{Q}\mb{\Lambda}\mb{Q}^{\top}
\end{align}
be the eigendecomposition of $(\mb{I}-\mb{H})^{-\frac{1}{2}}(\mb{G}-\mb{H})^{\frac{1}{2}}(\mb{I}-\mb{H})^{-\frac{1}{2}}$, where $\mb{Q}$ is an orthogonal matrix and $\mb{\Lambda}$ is a diagonal matrix. Here we can take the square root of $\mb{I}-\mb{H}$ because it is a positive-semidefinite matrix according to Remark~\ref{rem8}. By Fact \ref{fact2}, $s$ is the largest eigenvalue of $({\bf I-H})^{-\frac{1}{2}}({\bf G-H})({\bf I-H})^{-\frac{1}{2}}$, hence we can assume without loss of generality that $\Lambda_{1,1}=s$. For each $\epsilon>0$ define the $L\times L$ matrices
\begin{align}
\mb{D}_{\epsilon}=\left(
                    \begin{array}{cc}
                      \epsilon & \mb{0} \\
                      \mb{0} & \mb{0} \\
                    \end{array}
                  \right)
\end{align}
and
\begin{align}
{\mb\Delta}_{\epsilon}=(\mb{I}-\mb{H})^{-\frac{1}{2}}\mb{Q}\mb{D}_{\epsilon}\mb{Q}^{\top}(\mb{I}-\mb{H})^{-\frac{1}{2}}.
\end{align}
Choose $\mb{V}_{\epsilon}$ to be a random $L$-vector such that $\mb{V}_{\epsilon}$ and $\mb{X}$ are jointly Gaussian, $\mb{V}_{\epsilon}-\mb{X}-(\mb{Y},\mb{Z})$, and
\begin{align}\label{eq219}
\mb{\Sigma}_{\mb{X}|\mb{V}_{\epsilon}}=\mb{\Sigma}_{\mb{X}}^{\frac{1}{2}}(\mb{I}-\mb{\Delta}_{\epsilon})
\mb{\Sigma}_{\mb{X}}^{\frac{1}{2}}.
\end{align}
This determines the joint distribution (up to a shift and a linear transform of $\mb{V}_{\epsilon}$, which are irrelevant), since the unconditional covariance of $\mb{X}$ is given in the problem statement.
Then, observe that (see \eqref{e_237}-\eqref{e237}):
\begin{figure*}[b]%\usepackage{dblfloatfix}
\hrulefill
\begin{align}
{\mb \Sigma}_{\mb{Y}|\mb{V}_{\epsilon}}
&=\mb{\Sigma_{\mb Y|X}}
+\mathbb{E}\left[(\mathbb{E}[\mb{Y|X}]-\mathbb{E}[\mb{Y}|\mb{V}_{\epsilon}])
(\mathbb{E}[\mb{Y}|\mb{X}]-\mathbb{E}[\mb{Y}|\mb{V}_{\epsilon}])^{\top}|\mb{V}_{\epsilon}\right]
\label{e_237}
\\
&={\bf\Sigma_{\bf Y|X}}
+\mathbb{E}\left[{\bf\Sigma_{\bf YX}}{\bf\Sigma_{\bf X}}^{-1}
({\bf X}-\mathbb{E}[{\bf X}|{\bf V}_{\epsilon}])
({\bf X}-\mathbb{E}[{\bf X}|{\bf V}_{\epsilon}])^{\top}
{\bf\Sigma_{X}}^{-1}{\bf\Sigma_{\bf XY}}|{\bf V}_{\epsilon}
\right]
\\
&={\bf\Sigma_{Y|X}}+{\bf \Sigma_{\bf YX}}{\bf \Sigma_{\bf X}}^{-1}{\bf \Sigma}_{\mb{X}|\mb{V}_{\epsilon}}
{\bf \Sigma_{\bf X}}^{-1}{\bf \Sigma_{\bf XY}}
\\
&={\bf\Sigma_{Y}}-{\bf\Sigma_{YX}\Sigma_{\mb X}}^{-\frac{1}{2}}{\mb \Delta}_{\epsilon}{\mb\Sigma_{X}}^{-\frac{1}{2}}{\mb\Sigma_{\mb{XY}}}\label{e237}
\end{align}
\end{figure*}
where the last step uses \eqref{eq219}. Hence,
\begin{align}
&I(\mb{V}_{\epsilon};\mb{Y})
\nonumber\\
&=
\frac{1}{2}\log\frac{|\bf \Sigma_Y|}{|{\bf \Sigma}_{\mb{Y}|\mb{V}_{\epsilon}}|}
\\
&=-\frac{1}{2}\log|\mb{I}-{\bf\Sigma_Y}^{-\frac{1}{2}}{\bf\Sigma_{YX}\Sigma_{\mb X}}^{-\frac{1}{2}}{\mb \Delta}_{\epsilon}{\mb\Sigma_{X}}^{-\frac{1}{2}}{\mb\Sigma_{\mb{XY}}}{\bf\Sigma_Y}^{-\frac{1}{2}}|
\\
&=-\frac{1}{2}\log|{\bf I-G\Delta}_{\epsilon}|\label{eq227}
\\
&=\frac{\log e}{2}{\rm tr}({\bf G\Delta}_{\epsilon})+O(\epsilon^2)\label{eq228}
\end{align}
where \eqref{eq227} uses Fact \ref{fact2} (or the Sylvester determinant identity) and \eqref{eq228} uses Fact \ref{fact3}. By the same token, we have shown
\begin{align}
I(\mb{V}_{\epsilon};\mb{X})=\frac{\log e}{2}{\rm tr}({\bf \Delta}_{\epsilon})+O(\epsilon^2)
\end{align}
and
\begin{align}
I(\mb{V}_{\epsilon};\mb{Z})=\frac{\log e}{2}{\rm tr}({\bf H\Delta}_{\epsilon})+O(\epsilon^2).
\end{align}
Therefore,
\begin{align}
&\lim_{\epsilon\downarrow0}\frac{I(\mb{V}_{\epsilon};\mb{Y})-I(\mb{V}_{\epsilon};\mb{Z})}
{I(\mb{V}_{\epsilon};\mb{X})-I(\mb{V}_{\epsilon};\mb{Z})}
\nonumber\\
&=\lim_{\epsilon\downarrow0}\frac{{\rm tr}({\bf (G-H)\Delta}_{\epsilon})}{{\rm tr}({\bf (I-H)\Delta}_{\epsilon})}
\\
&=\lim_{\epsilon\downarrow0}\frac{\tr\left(({\bf G-H})({\bf I-H})^{-\frac{1}{2}}\mb{Q}
{\bf D}_{\epsilon}\mb{Q}^{\top}({\bf I-H})^{-\frac{1}{2}}\right)}
{\tr(\mb{D}_{\epsilon})}
\\
&=\lim_{\epsilon\downarrow0}\frac{\tr(\mb{\Lambda}\mb{D}_{\epsilon})}{\tr(\mb{D}_{\epsilon})}
\\
&=\Lambda_{1,1}
\\
&=s,
\end{align}
Hence by \eqref{e215} we have shown that $s^*_{\mb{Z}}(\mb{X};\mb{Y})\ge s=s^+$.

Conversely, to show $s^*_{\mb{Z}}(\mb{X};\mb{Y})\le s^+$, we may assume without loss of generality that $s<1$ since Remark~\ref{rem8} implies that $s\le1$ and when $s=1$ the claim is trivially true. We have remarked that $s$ is the largest eigenvalue of $({\bf I-H})^{-\frac{1}{2}}({\bf G-H})({\bf I-H})^{-\frac{1}{2}}$, hence
\begin{align}
({\bf I-H})^{-\frac{1}{2}}({\bf G-H})({\bf I-H})^{-\frac{1}{2}}\preceq s\mb{I},
\end{align}
which implies
\begin{align}\label{eq233}
{\bf I-G}\succeq (1-s)({\bf I-H}).
\end{align}
Now define $\hat{\mb{H}}:=\mb{I}-\frac{1}{1-s}(\bf I-G)$, then
\begin{align}\label{eq234}
{\bf I-G}=(1-s)({\bf I-\hat{H}}).
\end{align}
From \eqref{eq233} and \eqref{eq234} it is clear that
\begin{align}\label{eq235}
\bf \hat{H}\preceq H.
\end{align}
By \eqref{eq235}, we can find a Gaussian $L$-vector $\mb{W}$ independent of $(\bf X,Y,Z)$ and define
\begin{align}
\bf\hat{Z}=Z+W
\end{align}
such that
\begin{align}
\mb{\hat{H}}&=\mb{\Sigma}^{-1/2}_{\mb{X}}\mb{\Sigma}_{\mb{X}\mb{Z}}
\mb{\Sigma}^{-1}_{\mb{\hat{Z}}}\mb{\Sigma}_{\mb{Z}\mb{X}}\mb{\Sigma}^{-1/2}_{\mb{X}}.
\end{align}
Since $\bf X\perp W$, we see that
\begin{align}
\mb{\hat{H}}&=\mb{\Sigma}^{-1/2}_{\mb{X}}\mb{\Sigma}_{\mb{X}\mb{\hat{Z}}}
\mb{\Sigma}^{-1}_{\mb{\hat{Z}}}\mb{\Sigma}_{\mb{\hat{Z}}\mb{X}}\mb{\Sigma}^{-1/2}_{\mb{X}},
\end{align}
which agrees with the definition \eqref{defn}, i.e. $\mb{\hat{H}}$ is the corresponding matrix for the source $\bf (X,Y,\hat{Z})$.
A noisier observation for the eavesdropper is advantageous for key generation, hence $\eta_{\mb{\hat{Z}}}(\mb{X};\mb{Y})\ge\eta_{\mb{Z}}(\mb{X};\mb{Y})$, and so $s^*_{\mb{\hat{Z}}}(\mb{X};\mb{Y})\ge s^*_{\mb{Z}}(\mb{X};\mb{Y})$. Moreover from \eqref{eq234} we see that $\hat{\mb{H}}$ commutes with $\mb{G}$, so that we can apply Lemma \ref{lem3} to find invertible linear transforms $\mathbf{X}\mapsto \bar{\mathbf{X}}$, $\mathbf{Y}\mapsto \bar{\mathbf{Y}}$, $\mathbf{\hat{Z}}\mapsto \bar{\mathbf{Z}}$ such that $\bf (\bar{X},\bar{Y},\bar{Z})$ is a product source in the sense of \eqref{assump1} and \eqref{assump2}. Furthermore, from the proof of Lemma \ref{lem3} one sees that
\begin{align}
1-\rho_{\bar{X}_i\bar{Y}_i}=(1-s)(1-\rho_{\bar{X}_i\bar{Z}_i})
\end{align}
for $i=1,\dots,L$. Hence by \eqref{eq16},
\begin{align}
s^*_{\mb{\hat{Z}}}(\mb{X};\mb{Y})
&=s^*_{\mb{\bar{Z}}}(\mb{\bar{X}};\mb{\bar{Y}})
\\
&=s^*_{\bar{Z}_i}(\bar{X}_i;\bar{Y}_i)
\\
&=\left(\frac{\rho_{\bar{X}_i\bar{Y}_i}-\rho_{\bar{X}_i\bar{Z}_i}}{1-\rho_{\bar{X}_i\bar{Z}_i}}\right)^+
\\
&=s^+,
\end{align}
and we can conclude that
\begin{align}
s^*_{\mb{Z}}(\mb{X};\mb{Y})\le s^*_{\mb{\hat{Z}}}(\mb{X};\mb{Y})\le s^+.
\end{align}

In summary we have shown that $s^*_{\mb{Z}}(\mb{X};\mb{Y})=s^+$, or equivalently
\begin{align}
\eta_{\mb{Z}}(\mb{X};\mb{Y})=\frac{s^+}{1-s^+}=\left(\frac{s}{1-s}\right)^+=\lambda^+_{\max}((\mb{G}-\mb{H})(\mb{I}-\mb{G})^{-1}),
\end{align}
 as desired.

\section{Proof of Lemma \ref{lemma3}}\label{app_lem3}
From Jensen's inequality we have
\begin{align}
\ln\mathbb{E}e^{t\eta}\ge\mathbb{E}\ln e^{t\eta}=t\mathbb{E}\eta.
\end{align}
The proof of the other part of the bound in (\ref{el1}) is essentially based on uniform integrality of $\{e^{t\eta}\}_{\rho\in[0,1-\delta]}$. Without loss of generality we can assume that $U,X$ are zero mean with unit variance. Also it suffices to consider only the case of $\rho\ge0$ since otherwise the correlation coefficient between $-U$ and $X$ is $-\rho>0$ but the distribution of $\imath_{-U;X}(-U;X)$ is the same as that of $\eta$. Now $N:=\frac{X-\rho U}{\sqrt{1-\rho^2}}$ is zero mean, with unit variance, and independent of $U$. Note that
\begin{align}
|\eta|&=\left|\frac{1}{2}\log\frac{1}{1-\rho^2}-\frac{1}{2}\log e\left(\frac{\rho^2U^2+\rho^2X^2}{1-\rho^2}-\frac{2\rho UX}{1-\rho^2}\right)\right|
\\
&\le \frac{1}{2}\log\frac{1}{1-\rho^2}+\frac{1}{2}\log e\cdot\frac{\rho^2U^2+\rho^2X^2+\rho(U^2+X^2)}{1-\rho^2}
\\
&\le \frac{1}{2}\log\frac{1}{1-\rho^2}+\frac{1}{2}\log e\cdot\frac{\rho(U^2+X^2)}{1-\rho}
\\
&\le \frac{1}{2}\log\frac{1}{1-\rho^2}+\frac{\rho\log e}{2\delta}(U^2+X^2)
\\
&= \frac{1}{2}\log\frac{1}{1-\rho^2}+\frac{\rho\log e}{2\delta}[U^2+(\sqrt{1-\rho^2}N+\rho U)^2]
\\
&\le \frac{1}{2}\log\frac{1}{1-\rho^2}+\frac{3\rho\log e}{2\delta}(U^2+N^2)\label{eq136}
\\
&\le \frac{1}{2}\log\frac{1}{1-(1-\delta)^2}+\frac{3\log e}{2\delta}(U^2+N^2).\label{e136}
\end{align}
It is easy to show that for any $\lambda<\frac{1}{4}$, $\mathbb{E}e^{2\lambda U^2}=\mathbb{E}e^{2\lambda N^2}$ is finite, and hence
\begin{align}
\mathbb{E}[e^{\lambda U^2}e^{\lambda N^2}]
&\le \sqrt{\mathbb{E}e^{2\lambda U^2}\mathbb{E}e^{2\lambda N^2}}
\\
&<\infty.\label{e138}
\end{align}
Let $\xi$ be the random variable in (\ref{e136}), whose distribution does not depend on $\rho$. By (\ref{e138}), $\mathbb{E}e^{t\xi}<\infty$ for all $0<t<\frac{\delta}{6\log e}$. Now for each $\delta>0$,
\begin{align}
&\lim_{\Delta\to \infty}\sup_{\rho\in[0,1-\delta],t\in(0,\frac{\delta}{7\log e}]}\mathbb{E}1_{|\eta|\ge \Delta}e^{t\eta}
\nonumber\\
&\le
\lim_{\Delta\to \infty}\sup_{\rho\in[0,1-\delta],t\in(0,\frac{\delta}{7\log e}]}\mathbb{E}1_{|\eta|\ge \Delta}e^{t|\eta|}
\\
&\le
\lim_{\Delta\to \infty}\sup_{t\in(0,\frac{\delta}{7\log e}]}\mathbb{E}1_{\xi\ge \Delta}e^{t\xi}
\\
&\le \lim_{\Delta\to \infty}\mathbb{E}1_{\xi\ge \Delta}\exp\left(\frac{\delta\xi}{7\log e}\right)
\\
&=0
\end{align}
where the last step follows from bounded convergence theorem (or dominated convergence theorem). Then there exists $\Delta_0>0$ large enough such that
\begin{align}
\sup_{\rho\in[0,1-\delta],t\in(0,\frac{\delta}{7\log e}]}\mathbb{E}1_{|\eta|\ge \Delta}e^{t\eta}&<\frac{\epsilon}{4},
\\
\mathbb{P}[\xi<\Delta_0]&>\frac{1}{2}.
\end{align}
for $\Delta\ge\Delta_0$. Now observe that from (\ref{eq136}), there exists a r.v. $\zeta=C_1(U^2+N^2)+C_2$ such that $|\eta|<\rho\zeta$ whenever $\rho<\frac{1}{2}$, where $C_1,C_2$ are constants depending only on $\delta$. Then,
\begin{align}
\sup_{t<\frac{1}{2}}\sup_{\rho<\delta_0}\frac{\ln\mathbb{E}e^{t\eta}}
{t}&\le
\sup_{t<\frac{1}{2}}\sup_{\rho<\delta_0}\frac{\ln\mathbb{E}e^{\rho t\zeta}}{t}
\\
&=\sup_{t<\frac{1}{2}}\frac{\ln\mathbb{E}e^{\delta_0 t\zeta}}{t}
\\
&=2\ln\mathbb{E}e^{\frac{\delta_0\zeta}{2}}\label{e148}
\\
&\to 0,\quad \delta_0\to 0,
\end{align}
where (\ref{e148}) follows from convexity of the cumulant generating function. Thus, we can pick $\delta_0$ small enough such that
\begin{align}\label{e150}
\sup_{t<\frac{1}{2}}\sup_{\rho<\delta_0}\frac{\ln\mathbb{E}e^{t\eta}}
{t}\le\epsilon.
\end{align}
On the other hand, for $\rho\ge\delta_0$ we have
\begin{align}
&\sup_{\delta_0\le\rho<1-\delta}\frac{\ln\mathbb{E}e^{t\eta}}{t\mathbb{E}\eta}
\nonumber\\
&=\sup_{\delta_0\le\rho<1-\delta}\frac{\ln(\mathbb{E}1_{|\eta|\ge \Delta_0}e^{t\eta}+\mathbb{E}1_{|\eta|<\Delta_0}e^{t\eta})}{t\mathbb{E}\eta}
\\
&\le \sup_{\delta_0\le\rho<1-\delta}\frac{\ln\mathbb{E}1_{|\eta|< \Delta_0}e^{t\eta}}{t\mathbb{E}\eta}
\left(1+\frac{\mathbb{E}1_{|\eta|\ge \Delta_0}e^{t\eta}}{\mathbb{E}1_{|\eta|< \Delta_0}e^{t\eta}}\right)
\\
&\le \sup_{\delta_0\le\rho<1-\delta}\left\{\frac{\ln[(\mathbb{E}t\eta+1)\cdot\beta_t
]}{t\mathbb{E}\eta}\left(1+\frac{\epsilon/4}{e^{-t\Delta_0}\mathbb{
P}[|\eta|<\Delta_0]}\right)\right\}
\\
&\le\sup_{\delta_0\le\rho<1-\delta}\left\{
\frac{\mathbb{E}t\eta+\ln\beta_t}{t\mathbb{E}\eta}
(1+\frac{\epsilon}{2}e^{t\Delta_0})\right\}
\\
&\le\sup_{\delta_0\le\rho<1-\delta}\left\{
\left(1+\frac{\ln\beta_t}{t\cdot\frac{1}{2}\log\frac{1}{1-\delta_0^2}}
(1+\frac{\epsilon}{2}e^{t\Delta_0})\right)\right\}
\\
&\to 1+\epsilon/2,\quad t\to0.\label{e156}
\end{align}
where we have defined $\beta_t:=\max\{\frac{e^{t\Delta_0}}{1+t\Delta_0},
\frac{e^{-t\Delta_0}}{1-t\Delta_0}\}$, and used the fact that $\beta_t=1+O(t^2)$ when $t\to 0$. Finally, in view of (\ref{e150}) and (\ref{e156}), there exist $t<\frac{1}{2}$ small enough such that for each $\rho\in[0,1-\delta]$, either $\ln\mathbb{E}e^{t\eta}<t\epsilon$ or $\ln\mathbb{E}e^{t\eta}<(1+\epsilon)\mathbb{E}t\eta$ hold.

\section{Proof of Lemma \ref{lem4}}\label{appf}
\begin{description}
  \item[(a)] The asymptotic equivalences have been remarked earlier in \eqref{e116}, so we only have to bound the eigenvalues. From \cite[Lemma~4.1]{gray2006toeplitz} we have
\begin{align}
0&<\min_{\omega\in[0,2\pi)}S_X(\omega)
\nonumber\\
&\le\lambda_{\min}({\bf\Sigma_X})
\nonumber\\
&\le\lambda_{\max}({\bf\Sigma_X})
\nonumber\\
&\le \max_{\omega\in[0,2\pi)}S_X(\omega);
\end{align}
from \eqref{def2} the eigenvalues of $\bf \Sigma_{\tilde{X}}$ are $\{S_X(\frac{2\pi k}{n})\}_{k=1}^n$, which are also bounded between $\min_{\omega\in[0,2\pi)}S_X(\omega)$ and $\max_{\omega\in[0,2\pi)}S_X(\omega)$.
Similarly, the eigenvalues of $\bf\Sigma_{\tilde{Y}}$ and $\bf\Sigma_Y$ are bounded between $\min_{\omega\in[0,2\pi)}S_Y(\omega)$ and $\max_{\omega\in[0,2\pi)}S_Y(\omega)$; and the eigenvalues of $\bf\Sigma_{\tilde{Z}}$ and $\bf\Sigma_Z$ are bounded between $\min_{\omega\in[0,2\pi)}S_Z(\omega)$ and $\max_{\omega\in[0,2\pi)}S_Z(\omega)$.

  \item[(b)] We first show that $\bf \Sigma_{X|U}\sim\Sigma_{\tilde{X}|\hat{U}}$. Let $\mb{R}$ be the diagonal matrix whose $(i,i)$ entry is $\rho^{(n)}_i$. Clearly both $\bf\Sigma_{X|U}$ and $\bf\Sigma_{\tilde{X}|\hat{U}}$ depend only on $\bf R$ and $\bf\Sigma_{\hat{X}}$, and do not depend on the scaling of $\bf \hat{U}$. However, to compute $\bf\Sigma_{\tilde{X}|\hat{U}}$, it is convenient to specify $P_{\bf\hat{U}|\hat{X}}$ via the following random transformation:
\begin{align}
{\bf\hat{U}}=({\bf I-R}^2)^{\frac{1}{2}}{\bf W+R\hat{X}},
\end{align}
where $\bf W$ is a zero mean Gaussian vector with covariance matrix $\bf\Sigma_{\hat{X}}$ and independent of $\bf\hat{X}$. Then the conditional covariance matrices can be expressed as
\begin{align}
{\bf\Sigma_{\tilde{X}|\hat{U}}}&=
{\bf\Sigma_{\tilde{X}}-\Sigma_{\tilde{X}\hat{U}}\Sigma_{\hat{U}}}^{-1}{\bf\Sigma_{\hat{U}\tilde{X}}}
\\
&={\bf\Sigma_{\tilde{X}}-\Sigma_{\tilde{X}}QR[(I-R}^2)^{\frac{1}{2}}{\bf\Sigma_{\hat{X}}(I-R}^2)^{\frac{1}{2}}
\nonumber\\
&\quad+{\bf RQ^{\top}\Sigma_{\tilde{X}}QR}]^{-1}\bf RQ^{\top}\Sigma_{\tilde{X}}.
\end{align}
and
\begin{align}
{\bf\Sigma_{X|U}}&=
{\bf\Sigma_{X}-\Sigma_{XU}\Sigma_U}^{-1}{\bf\Sigma_{UX}}
\\
&={\bf\Sigma_{X}-\Sigma_{X}QR[(I-R}^2)^{\frac{1}{2}}{\bf\Sigma_{\hat{X}}(I-R}^2)^{\frac{1}{2}}
\nonumber\\
&\quad+{\bf RQ^{\top}\Sigma_{X}QR}]^{-1}\bf RQ^{\top}\Sigma_{X}.
\end{align}
It is easy to see that the smallest eigenvalue of ${\bf(I-R}^2)^{\frac{1}{2}}{\bf\Sigma_{\hat{X}}(I-R}^2)^{\frac{1}{2}}$ is lower bounded by $\min_{\omega\in[0,2\pi)}S_X(\omega)(1-\max_{\omega\in[0,2\pi)}\rho_{UX}^2(\omega))$ which is positive due to \eqref{e132}.
Therefore, Fact~\ref{fact4} and Part (a) imply the asymptotic equivalence $\bf \Sigma_{X|U}\sim\Sigma_{\tilde{X}|\hat{U}}$.

Next, from the Markov chains $\bf U-X-Y$ and $\bf \hat{U}-\tilde{X}-\tilde{Y}$, we can show that (similar to the derivations in \eqref{e237})
\begin{align}
{\bf\Sigma_{Y|U}}
&={\bf\Sigma_{Y|X}+\Sigma_{YX}\Sigma_{X}}^{-1}{\bf\Sigma_{X|U}\Sigma_{X}}^{-1}{\bf\Sigma_{XY}}
\nonumber\\
&={\bf\Sigma_Y}-{\bf\Sigma_{YX}\Sigma_X}^{-1}{\bf\Sigma_{XY}}
\nonumber\\
&\quad+{\bf\Sigma_{YX}\Sigma_{X}}^{-1}{\bf\Sigma_{X|U}\Sigma_{X}}^{-1}{\bf\Sigma_{XY}}\label{e284}
\end{align}
and
\begin{align}
{\bf\Sigma_{\tilde{Y}|\hat{U}}}
&={\bf\Sigma_{\tilde{Y}}}-{\bf\Sigma_{\tilde{Y}\tilde{X}}\Sigma_{\tilde{X}}}^{-1}{\bf\Sigma_{\tilde{X}\tilde{Y}}}
\nonumber\\
&\quad+{\bf\Sigma_{\tilde{Y}\tilde{X}}\Sigma_{\tilde{X}}}^{-1}{\bf\Sigma_{\tilde{X}|\hat{U}}
\Sigma_{\tilde{X}}}^{-1}{\bf\Sigma_{\tilde{X}\tilde{Y}}}.
\end{align}
Therefore \eqref{e116}, $\bf\Sigma_{X|U}\sim\Sigma_{\tilde{X}|\hat{U}}$, and Part (a) immediately establish the relation ${\bf\Sigma_{Y|U}}\sim{\bf\Sigma_{\tilde{Y}|\hat{U}}}$.

Note that \eqref{e284} can be written as
\begin{align}
{\bf\Sigma_{Y|U}}
={\bf\Sigma_Y}^{\frac{1}{2}}
[{\bf I}-{\bf A}({\bf I-\Sigma_X}^{-\frac{1}{2}}{\bf \Sigma_{X|U}\Sigma_X}^{-\frac{1}{2}}){\bf A}^{\top}]{\bf \Sigma_Y}^{\frac{1}{2}},
\end{align}
where we have defined ${\bf A}={\bf \Sigma_Y}^{-\frac{1}{2}}{\bf\Sigma_{YX}}{\bf\Sigma_X}^{-\frac{1}{2}}$.
From the result of Part (a) we see that
\begin{align}
\lambda_{\max}({\bf I}-{\bf \Sigma_X}^{-\frac{1}{2}}{\bf\Sigma_{X|U}}{\bf\Sigma_X}^{-\frac{1}{2}})<1-\delta
\end{align}
for some $\delta>0$ which is independent of $n$. However the positive-semidefiniteness of the covariance matrix of $\bf(X^{\top},Y^{\top})$ implies that the largest singular value $\sigma_{\max}({\bf A})\le 1$, which in turn gives
\begin{align}
\lambda_{\max}({\bf A}({\bf I}-{\bf \Sigma_X}^{-\frac{1}{2}}{\bf\Sigma_{X|U}}{\bf\Sigma_X}^{-\frac{1}{2}}){\bf A}^{\top})<1-\delta.
\end{align}
Therefore we have the uniform lower bound
\begin{align}
&\lambda_{\min}({\bf\Sigma_{Y|U}})
\nonumber\\
&\ge\min_{\omega\in[0,2\pi)}S_Y(\omega)
\\
&\quad(1-\lambda_{\max}({\bf A}({\bf I}-{\bf \Sigma_X}^{-\frac{1}{2}}{\bf\Sigma_{X|U}}{\bf\Sigma_X}^{-\frac{1}{2}}){\bf A}^{\top}))
\\
&>\delta\min_{\omega\in[0,2\pi)}S_Y(\omega),\quad\forall n>0.
\end{align}
A similar uniform lower bound can be obtained for $\bf\Sigma_{\tilde{Y}|\hat{U}}$.
The relation ${\bf\Sigma_{Z|U}}\sim{\bf\Sigma_{\tilde{Z}|\hat{U}}}$ and the uniform lower boundedness of their eigenvalues can be shown in the exactly same way since the roles of $\mathbb{Y}$ and $\mathbb{Z}$ are equal for this problem.
\end{description}

\section{Converse of Theorem \ref{thm5}}\label{app_conv}
The first step towards the converse proof is to bound the key rate and the transmission rate with multi-letter expressions. This part is similar to the initial steps in the converse proof of key capacity of memoryless sources, c.f. \cite{csiszar2000common}.

Consider
\begin{align}
\log|\mathcal{K}|&=H(K|W,Z^n)+\nu_n
\\
&\le H(K)+\nu_n\label{e254}
\\
&\le H(K)-H(K|Y^n,W)+n\gamma_n+\nu_n\label{e152}
\\
&= I(K;Y^n,W)+n\gamma_n+\nu_n
\\
&\le I(K;Y^n,W)-I(K;Z^n,W)+n\gamma_n+2\nu_n\label{e154}
\\
&= I(K;Y^n|W)-I(K;Z^n|W)+n\gamma_n+2\nu_n,\label{e155}
\end{align}
where \eqref{e254} and \eqref{e154} are from the definition of $\nu_n$ and (\ref{e152}) is from Fano's inequality, with $\gamma_n:=\frac{1}{n}[\epsilon_n\log|\mathcal{K}_1|+h(\epsilon_n)]$.

As for the transmission rate, note that
\begin{align}
\log|\mathcal{W}|&\ge H(W)
\\
&\ge H(W|Y^n)-H(W,K|X^n)
\\
&\ge H(W|Y^n)+H(K|W,Y^n)-n\gamma_n-H(W,K|X^n)\label{e158}
\\
&= H(K,W|Y^n)-n\gamma_n-H(W,K|X^n)
\\
&= I(K,W;X^n)-I(K,W;Y^n)-n\gamma_n,\label{e162}
\end{align}
where (\ref{e158}) used Fano's inequality.

Now suppose $(R,r)$ is achievable, where $r>0,R>0$. We identify $K$ and $W$ in (\ref{e155}), (\ref{e162}) with $V,U$ respectively, and then apply Fact \ref{fact1}. Also notice that $\lim_{n\to 0}\gamma_n=\lim_{n\to \infty}\nu_n=0$. These imply the existence of a sequence of conditional Gaussian distributions $P_{U^n|X^n}$ such that
\begin{align}
r&\ge\lim_{n\to\infty}\frac{1}{n}[I(X^n;U^n)-I(Y^n;U^n)];\\
R&\le\lim_{n\to\infty}\frac{1}{n}[I(Y^n;U^n)-I(Z^n;U^n)].
\end{align}
As in Section \ref{sec5}, let $\tilde{\mb{X}},\tilde{\mb{Y}}$ and $\tilde{\mb{Z}}$ be jointly Gaussian vectors with circulant covariance matrices defined in \eqref{ecirc}; and $\hat{\mb{X}},\hat{\mb{Y}},\hat{\mb{Z}}$ be the result of applying the linear transforms in \eqref{e123}-\eqref{e125}. Then
\begin{align}
r&\ge\lim_{n\to\infty}\frac{1}{n}[I(\hat{\mb{X}};
\hat{\mb{U}})-I(\hat{\mb{Y}};\hat{\mb{U}})];\\
R&\le\lim_{n\to\infty}\frac{1}{n}[I(\hat{\mb{Y}};
\hat{\mb{U}})-I(\hat{\mb{Z}};\hat{\mb{U}})].
\end{align}
For $x>0$ define the decreasing functions:
\begin{align}
f(x)&=\frac{1}{4\pi}\int_{\beta(\omega)>x}
\log\frac{\beta(\omega)(x+1)}{(\beta(\omega)
+1)x}{\rm d}\omega,
\\
g(x)&=\frac{1}{4\pi}\int_{\beta(\omega)>x}
\log\frac{\beta(\omega)+1}{x+1}{\rm d}\omega,
\\
f_n(x)&:=\frac{1}{2n}\sum_{i:\beta^{(n)}_i>x}
\log\frac{\beta^{(n)}_i(x+1)}{(\beta^{(n)}_i+1)x},
\\
g_n(x)&:=\frac{1}{2n}
\sum_{i:\beta^{(n)}_i>x}\log\frac{\beta^{(n)}_i+1}{x+1},
\end{align}
where
\begin{align}
\beta^{(n)}_i:=\frac{\rho_{\hat{X}^{(n)}_i\hat{Y}^{(n)}_i}^2-\rho_{\hat{X}^{(n)}_i\hat{Z}^{(n)}_i}^2}
{1-\rho_{\hat{X}^{(n)}_i\hat{Y}^{(n)}_i}^2}.
\end{align}

The empirical distribution of $\{\beta^{(n)}_i\}_{i=1}^n$ converges weakly to the distribution of $\beta(W)$ when $W$ is uniformly distributed on $[0,2\pi)$, which means that for any $x>0$ it holds that
\begin{align}
\lim_{n\to\infty}f_n(x)=&f(x),\label{e296}
\\
\lim_{n\to\infty}g_n(x)=&g(x).
\end{align}

By Theorem \ref{thm4}, there is a sequence $\{\mu_n\}$ such that $I(\hat{\mb{X}};
\hat{\mb{U}})-I(\hat{\mb{Y}};\hat{\mb{U}})\ge f_n(\mu_n)$ and $I(\hat{\mb{Y}};
\hat{\mb{U}})-I(\hat{\mb{Z}};\hat{\mb{U}})\le g_n(\mu_n)$, and so
\begin{align}
r&\ge\limsup_{n\to\infty}f_n(\mu_n),\label{e126}
\\
R&\le\liminf_{n\to\infty}g_n(\mu_n),\label{e127}
\end{align}
Define $\mu:=\limsup_{n\to\infty}\mu_n$. We observe that $\mu_n$ is bounded away from $0$ and $+\infty$: suppose on the contrary that it is not bounded away from $0$. Choose $\epsilon>0$ small enough such that $f(\epsilon)-\epsilon>r$ (which is possible since $\lim_{\epsilon\downarrow0}f(\epsilon)=+\infty$ by monotone convergence theorem), and there is a subsequence $\{\mu_{n_k}\}_{k=1}^{\infty}$ such that $\mu_{n_k}<\epsilon$ for all $k$. From monotonicity of $f_n$ we see that
\begin{align}
f_{n_k}(\mu_{n_k})\ge f_{n_k}(\epsilon)\to f(\epsilon),\quad k\to\infty.
\end{align}
This implies that $f_{n_k}(\mu_{n_k})>f(\epsilon)-\epsilon>r$ when $k$ is sufficiently large, which contradicts (\ref{e126}). Similarly we can also show that $\mu_n$ is upper bounded: if otherwise, we pick $M>0$ such that $g(M)<\frac{R}{2}$ (which is possible since $\lim_{x\to+\infty}g(x)=0$ by monotone convergence theorem), and choose a subsequence $\{\mu_{n_k}\}_{k=1}^{\infty}$ such that $\mu_{n_k}>M$ for all $k$. Then from monotonicity of $g_n$ we see that
\begin{align}
g_{n_k}(\mu_{n_k})\le g_{n_k}(\mu)\to g(M)<\frac{R}{2},\quad k\to\infty.
\end{align}
This implies that $g_{n_k}(\mu_{n_k})\le \frac{3}{4}R$ for $k$ large enough, which contradicts \eqref{e127}.
Thus, we may assume that $c<\mu_n<d$, for some $0<c<d$.

By differentiation it's easy to see that $f_n$ is $\frac{\log e}{2c(1+c)}$-Lipschitz on $[c,d]$. Now choose a new subsequence $\{\mu_{i_k}\}_{k=1}^{\infty}$ which converges to $\mu$. We have
\begin{align}
|f_{i_k}(\mu_{i_k})-f(\mu)|
&\le|f_{i_k}(\mu_{i_k})-f_{i_k}(\mu)|+|f_{i_k}(\mu)-f(\mu)|\label{eqq313}
\\
&\le\frac{\log e}{2c(1+c)}|\mu_{i_k}-\mu|+|f_{i_k}(\mu)-f(\mu)|
\\
&\to 0,\quad k\to\infty.\label{e280}
\end{align}
where \eqref{e280} used \eqref{e296}. Hence
\begin{align}\label{e128}
\liminf_{n\to\infty}f_n(\mu_n)\le\lim_{k\to\infty}f_{i_k}(\mu_{i_k})=f(\mu).
\end{align}
Similarly to \eqref{eqq313}-\eqref{e128}, we can also show that
\begin{align}\label{e129}
\limsup_{n\to\infty}g_n(\mu_n)\le\lim_{k\to\infty}g_{i_k}
(\mu_{i_k})=g(\mu).
\end{align}
%
%Then
%\begin{align}
%\frac{1}{2}\liminf_{n\to\infty}\frac{1}{n}\sum_{i:\beta^{(n)}_i>\mu_n}\log\frac{\beta^{(n)}_i(\mu_n+1)}{(\beta^{(n)}_i+1)\mu_n}
%\le& \frac{1}{4\pi}\int_{\beta(\omega)>\mu}\log\frac{\beta(\omega)(\mu+1)}{(\beta(\omega)
%+1)\mu}{\rm d}\omega,\\
%\frac{1}{2}\limsup_{n\to\infty}\frac{1}{n}\sum_{i:\beta^{(n)}_i>\mu_n}\log\frac{\beta^{(n)}_i+1}{\mu_n+1}
%\ge& \frac{1}{4\pi}\int_{\beta(\omega)>\mu}\log\frac{\beta(\omega)+1}{\mu+1}{\rm d}\omega.\label{e129}
%\end{align}
The proof is accomplished by combining (\ref{e126}), (\ref{e127}), (\ref{e128}) and (\ref{e129}).

\section{Review of Results on Toeplitz Approximation}\label{appI}
The asymptotic distribution of the eigenvalues of Toeplitz matrices can be described in terms of the ``equal distribution'' introduced by H. Weyl \cite{weyl1916gleichverteilung}.
\begin{defn}\cite{grenander1958toeplitz}
For each $n$ consider two sets of $n$ real numbers $\{a^{(n)}_i\}_{i=1}^n$ and $\{b^{(n)}_i\}_{i=1}^n$ satisfying
\begin{align}
A<a^{(n)}_i<B,~~A<b^{(n)}_i<B,\quad\forall 1\le i\le n,~n\ge1.
\end{align}
for some $A,B>0$. The sequences $\{a^{(n)}_i\}_{i=1}^n$ and $\{b^{(n)}_i\}_{i=1}^n$ are said to be \emph{asymptotically equally distributed} in $[A,B]$ if for any continuous function $F\colon[A,B]\to\mathbb{R}$, it holds that
\begin{align}\label{e106}
\lim_{n\to\infty}\frac{\sum_{i=1}^n[F(a^{(n)}_i)-F(b^{(n)}_i)]}{n}=0.
\end{align}
\end{defn}
Denote by $\mathbb{P}_{a^{(n)}}$ the empirical distribution of $\{a^{(n)}\}$. Then (\ref{e106}) can be expressed as
\begin{align}
\lim_{n\to\infty}[\mathbb{E}F(X_a)-\mathbb{E}F(X_b)]=0,
\end{align}
where the random variables $X_a$ and $X_b$ are distributed according to $\mathbb{P}_{a^{(n)}}$ and $\mathbb{P}_{b^{(n)}}$, respectively.
\begin{defn}
Consider the sets $\{a^{(n)}_i\}_{i=1}^n$ of real numbers from $[A,B]$. The empirical distribution $\mathbb{P}_{a^{(n)}}$ is said to \emph{converge weakly} to a measure $\mu$ on $[A,B]$ if for any continuous function $F\colon[A,B]\to\mathbb{R}$,
\begin{align}
\lim_{n\to\infty}\mathbb{E}F(X_a)=\mathbb{E}F(X_{\mu}),
\end{align}
where the random variables $X_a$ and $X_{\mu}$ are distributed according to $\mathbb{P}_{a^{(n)}}$ and $\mu$, respectively.
\end{defn}

\begin{defn}\label{asymeq}
\cite{gray2006toeplitz} We say $\mb{A}_n$ and $\mb{B}_n$ are \emph{asymptotically equivalent} (denoted as $\mathbf{A}_n\sim \mathbf{B}_n$) for two sequences of matrices $\{\mathbf{A}_n\}$ and $\{\mathbf{B}_n\}$ if
\begin{enumerate}
  \item $\mathbf{A}_n$ and $\mathbf{B}_n$ are uniformly bounded in $\ell_2$ operator norm, i.e. for some $M>0$,
  \begin{align}\label{c1}
  \|\mathbf{A}_n\|,\|\mathbf{B}_n\|\le M<\infty,\quad n=1,2,\dots;
  \end{align}
  \item $\mathbf{A}_n-\mathbf{B}_n$ converges to zero in the weak norm:
  \begin{align}\label{c2}
  \lim_{n\to \infty}|\mathbf{A}_n-\mathbf{B}_n|=0,
  \end{align}
  where $|\mathbf{A}|:=(\frac{1}{n}{\rm tr}[\mathbf{A}^{\dagger}\mathbf{\mathbf{A}}])^{1/2}$.
\end{enumerate}
\end{defn}
The key result we use in Section~\ref{sec5} is then expressed as:
\begin{fact}\cite[Lemma~4.6]{gray2006toeplitz}\label{fact6}
If $f$ is in the Wiener class then $\mb{T}_n(f)\sim \mb{C}_n(f)$, where the notations $\mb{T}_n$ and $\mb{C}_n$ are as in \eqref{e109} and \eqref{e110}.
\end{fact}
The following property will be useful later in proving asymptotic equivalence of Toeplitz matrices. The claim about square root matrices follows from \cite[Theorem 1]{gutierrez2008asymptotically} by particularizing the continuous function therein to the square root function, while all other claims are from \cite[Theorem 2.1]{gray2006toeplitz}.
\begin{fact}\label{fact4}
Sums and products of asymptotically equivalent matrices are asymptotically equivalent. If the smallest singular values of asymptotically equivalent matrices are uniformly lower bounded, then their inverses are also asymptotically equivalent. Moreover, square roots of asymptotically equivalent positive-semidefinite matrices are asymptotically equivalent.
\end{fact}

The relevance of asymptotically equivalent matrices to coding theorems lies in the following fact:
\begin{fact}\label{fact5}
\cite[Theorem 2.4]{gray2006toeplitz} Let $\mb{A}_n$ and $\mb{B}_n$ be asymptotically equivalent sequences of Hermitian matrices with eigenvalues inside the interval $[m,M]$. Then the eigenvalues of $\mb{A}_n$ and $\mb{B}_n$ are asymptotically equally distributed on $[m,M]$.
\end{fact}

\bibliographystyle{ieeetr}
\bibliography{ref2014}

\begin{IEEEbiographynophoto}{Jingbo Liu}
received the B.E. degree from Tsinghua University, Beijing, China in 2012 and the M.A. degree from Princeton University, Princeton, NJ, USA in 2014, both in electrical engineering. He is currently pursuing a Ph.D. degree at Princeton University. His research interests include signal processing, information theory, coding theory and the related fields. His undergraduate thesis on a topological viewpoint on non-convex sparse signal recovery received the best undergraduate thesis award at Tsinghua University (2012). He gave a semi-plenary presentation at the 2015 IEEE Int. Symposium on Information Theory, Hong-Kong, China.
\end{IEEEbiographynophoto}

\begin{IEEEbiographynophoto}{Paul Cuff} 
received the B.S. degree in electrical engineering from Brigham Young University, Provo, UT, in 2004 and the M.S. and Ph. D. degrees in electrical engineering from Stanford University in 2006 and 2009. Since 2009 he has been an Assistant Professor of Electrical Engineering at Princeton University.

As a graduate student, Dr. Cuff was awarded the ISIT 2008 Student Paper Award for his work titled ¡°Communication Requirements for Generating Correlated Random Variables¡± and was a recipient of the National Defense Science and Engineering Graduate Fellowship and the Numerical Technologies Fellowship. As faculty, he received the NSF Career Award in 2014 and the AFOSR Young Investigator Program Award in 2015.
\end{IEEEbiographynophoto}

\begin{IEEEbiographynophoto}
{Sergio Verd\'{u}} received the Telecommunications Engineering degree from the
Universitat Polit\`{e}cnica de Barcelona in 1980, and the Ph.D. degree in Electrical Engineering from the
University of Illinois at Urbana-Champaign in 1984. Since 1984 he has been a member of the faculty of
Princeton University, where he is the Eugene Higgins Professor of Electrical Engineering, and is a member
of the Program in Applied and Computational Mathematics.

Sergio Verd\'{u} is  the recipient of the 2007 Claude E. Shannon Award, and
the 2008 IEEE Richard W. Hamming Medal.
He is a member of both the National Academy of Engineering and the National Academy of Sciences.

Verd\'{u} is a recipient of several paper awards from the IEEE:
the 1992 Donald Fink Paper Award,
the 1998 and 2012 Information Theory  Paper Awards,
an Information Theory Golden Jubilee Paper Award,
the 2002 Leonard Abraham Prize Award,
the 2006 Joint Communications/Information Theory Paper Award,
and the 2009 Stephen O. Rice Prize from the IEEE Communications Society.
In 1998, Cambridge University Press published his book {\em Multiuser Detection,}
for which he received the 2000 Frederick E. Terman Award from the American Society for Engineering Education.
He was awarded a Doctorate Honoris Causa from the Universitat  Polit\`{e}cnica de Catalunya in 2005.

Sergio Verd\'{u} served as President of the IEEE Information Theory Society in 1997, and
on its Board of Governors (1988-1999, 2009-2014).
He has also served in various editorial capacities for the {\em IEEE Transactions on Information Theory}:
Associate Editor (Shannon Theory, 1990-1993; Book Reviews, 2002-2006),
Guest Editor of the Special Fiftieth Anniversary Commemorative Issue
(published by IEEE Press as ``Information Theory: Fifty years of discovery"),
and member of the Executive Editorial Board (2010-2013).
He is the founding Editor-in-Chief of {\em Foundations and Trends in Communications and Information Theory}.
Verd\'{u} is co-chair of the {\em 2016 IEEE International Symposium on Information Theory}, which will take place in his hometown.
\end{IEEEbiographynophoto}
\end{document}